%% file: covering-full.tex
\documentclass[final, 11pt]{article}

\newif\ifpublic
\publicfalse

\newif\ifarxiv
\arxivtrue

\ifpublic
     \usepackage[disable]{todonotes}
\else

	\usepackage[colorinlistoftodos]{todonotes}
\fi

\def\showauthornotes{1}

\def\showdraftbox{1}

\usepackage{mathpazo}
\usepackage{setspace}
\usepackage{amsmath,amssymb,amsfonts}
\usepackage{mathtools}
\usepackage{bbold}
\usepackage{fixltx2e}
\usepackage{epsfig}
\usepackage{latexsym,nicefrac,bbm}
\usepackage{xspace}
\usepackage{color,fancybox,graphicx,url,subfigure}
\usepackage{enumitem}
\usepackage{fullpage}
\usepackage{mathabx}
\usepackage{tabularx}
\usepackage{mdframed}
\usepackage{braket}

\usepackage[bookmarks,colorlinks,breaklinks]{hyperref}  
\hypersetup{linkcolor=blue,citecolor=blue,filecolor=blue,urlcolor=blue} 
\renewcommand{\eqref}[1]{\hyperref[#1]{(\ref*{#1})}}

\numberwithin{equation}{section}

\input{macros}

\newcommand{\lref}[2][]{\hyperref[#2]{#1~\ref*{#2}}}
\renewcommand{\eqref}[1]{\hyperref[#1]{(\ref*{#1})}}

\numberwithin{equation}{section}

\newcommand{\TSAT}{3\text{-}SAT}
\newcommand{\Var}{{\sf Var}}
 
\newcommand{\LC}{{\sc Label-Cover}}
\newcommand{\UG}{{\sc Unique-Games}}

\newcommand{\nae}{\mathsf{NAE}}
\newcommand{\csp}{{\sc CSP}}
\newcommand{\cov}{{\sc Covering}}

\newcommand{\ugc}{{\sc UGC}}

\newcommand{\err}{\mathsf{err}}
\newcommand{\Inf}{\mathsf{Inf}}
\newcommand{\Lab}{\mathsf{Lab}}
\newcommand{\term}{\mathsf{Term}}

\newcommand{\supp}{\mathsf{supp}}
\newcommand{\twoklin}{\mathsf{2k\text{-}LIN}}
\newcommand{\fourlin}{\mathsf{4\text{-}LIN}}

\newcommand{\calY}{\mathcal{Y}}
\newcommand{\calX}{\mathcal{X}}
\newcommand{\calU}{\mathcal{U}}

\newcommand{\calE}{\mathcal{E}}

\newcommand{\setq}{[q]}

\newcommand{\calT}{\mathcal{T}}
\newcommand{\calV}{\mathcal{V}}
\newcommand{\calP}{\mathcal{P}}

\DeclareMathOperator*{\E}{\mathbb{E}}

\newcommand{\threelin}{\mathsf{3\text{-}LIN}}
\newcommand{\threecnf}{\mathsf{3\text{-}CNF}}
\newcommand{\cover}[1]{\nu(#1)}
\renewcommand{\epsilon}{\varepsilon}

\newcommand{\ie}{i.\,e.}

\newcommand{\expref}[2]{\texorpdfstring{\hyperref[#2]{#1~\ref{#2}}}{#1~\ref{#2}}} 
\newcommand{\expeqref}[2]{\texorpdfstring{\hyperref[#2]{#1~\eqref{#2}}}{#1~\eqref{#2}}}

\date{}

\title{A characterization of hard-to-cover CSPs}
\author{
  Amey Bhangale\thanks{Department of Computer Science and Engineering, University of California, Riverside, CA, USA. Research supported by the NSF grant CCF-1253886.  This work was done when the author was a graduate student at Rutgers University, USA. Part of the work was done when the author was visiting TIFR. {Email : {\tt ameyrb@ucr.edu}}}
 \and 
Prahladh Harsha\thanks{Tata Institute of Fundamental Research,
   India.  Email : {\tt prahladh@tifr.res.in}}
\and
Girish Varma\thanks{International Institute of Information Technology, Hyderabad, India.  Supported by Google India under the Google India PhD Fellowship Award. This work was done when the author was a graduate student at TIFR, India. Email : {\tt girishrv@gmail.com}}
 }

\begin{document}
\begin{titlepage}
  
  \maketitle
  \thispagestyle{empty}
  \setcounter{page}{0}
\begin{abstract}
	We continue the study of
	the  
	\emph{covering complexity} of constraint
	satisfaction problems (CSPs) initiated by Guruswami, H{\aa}stad and
	Sudan [SIAM J. Comp. 2002] and Dinur and Kol [CCC'13]. 
	The covering number of a CSP instance
	$\Phi$ is the smallest number of assignments
	to the variables of $\Phi$, such that each constraint of $\Phi$ is
	satisfied by at least one of the assignments. We show the following
	results:
	\begin{enumerate}
		\item Assuming a covering 
		variant of the  
		Unique Games Conjecture,   
		introduced by
		Dinur and Kol, we show that for every non-odd predicate $P$ over
		any constant-size alphabet and every integer $K$, it is
		$\mathrm{NP}$-hard to approximate the covering number within a
		factor of $K$. This yields a complete characterization of CSPs
		over constant-size alphabets that are hard to cover.
		
		\item For a large class of predicates that are contained in the
		$\twoklin$ predicate, we show that it is quasi-$\mathrm{NP}$-hard to
		distinguish between instances 
		with 
		covering number at most 
		$2$   
		and
		those with covering number 
		at least $\Omega(\log\log n)$. This
		generalizes and improves the $\fourlin$ covering hardness result of Dinur and Kol.
	\end{enumerate}
\end{abstract}
\end{titlepage}


\section{Introduction}

One of the central (yet unresolved) questions in inapproximability is
the problem of coloring a (hyper)graph with as few 
colors as possible. A (hyper)graph $G=(V,E)$ is said to be $k$-colorable if there
exists a coloring $c:V\to [k] :=\{0,1,2,\dots,k-1\}$ of the vertices
such that no (hyper)edge of $G$ is monochromatic. The chromatic
number of a (hyper)graph, denoted by $\chi(G)$, is the smallest $k$
such that $G$ is $k$-colorable. It is known that computing $\chi(G)$
to within a multiplicative factor of $n^{1-\epsilon}$ on
an $n$-vertex graph $G$   
is 
$\mathrm{NP}$-hard for every $\epsilon \in (0,1)$~\cite{FeigeK1998, Zuckerman2007}.   
However, the
complexity of the following problem is not yet completely understood:
given a constant-colorable (hyper)graph, what is the minimum number of
colors required to color the vertices of the graph efficiently such
that every edge is non-monochromatic? The current best approximation
algorithms for this problem 
require at least $n^{\Omega(1)}$ colors
\cite{KrivelevichNS2001} 
while the hardness results are far from proving optimality of these
approximation algorithms. (See \expref{Sec.}{sec:recentcolors} for a
discussion on recent work in this area.) 

The notion of \emph{covering complexity} was introduced by Guruswami,
H{\aa}stad and Sudan~\cite{GuruswamiHS2002} and more formally by Dinur
and Kol~\cite{DinurK2013} to obtain a better understanding of the
complexity of this problem. Let $P$ be a predicate and $\Phi$ an
instance of a constraint satisfaction problem (CSP) over $n$
variables, where each constraint in $\Phi$ is a constraint of type $P$
over the $n$ variables and their negations. We will refer to such CSPs
as $P$-CSPs. The \emph{covering number} of $\Phi$, denoted by
$\cover{\Phi}$, is the smallest number of assignments to the variables
such that each constraint of $\Phi$ is satisfied by at least one of
the assignments, in which case we say that the set of assignments {\em
	covers} the instance $\Phi$. If $c$ assignments cover the instance
$\Phi$, we say that $\Phi$ is $c$-coverable or equivalently that the
set of assignments form a $c$-covering for $\Phi$. The covering number
is a generalization of the notion of chromatic number (to be more
precise, the logarithm of the the chromatic number) to all predicates
in the following sense. Let $G_\Phi$ be the underlying \emph{constraint
	(hyper)graph} of the instance $\Phi$ whose vertices are the
variables of the instance $\Phi$ and (hyper)edges are in one-to-one
correspondence with the constraints of $\Phi$. Suppose $P$ is the
not-all-equal predicate
$\nae$ and the instance $\Phi$ has no negations in any of its
constraints, then the covering number $\cover{\Phi}$ is exactly
$\lceil \log \chi(G_\Phi) \rceil$ where $G_\Phi$ is the underlying
constraint graph of the instance $\Phi$.

Cover-$P$ refers to the problem of finding the covering number
of a given $P$-CSP instance. Finding the exact covering number for
most interesting predicates $P$ is $\mathrm{NP}$-hard. We therefore study the
problem of approximating the covering number. In particular, we would
like to study the complexity of the following problem, denoted by
\cov-$P$-\csp$(c,s)$, for some $1\leq c < s \in \nat$: ``given a
$c$-coverable $P$-CSP instance $\Phi$, find an $s$-covering for
$\Phi$''. Similar problems have been studied for the Max-CSP setting:
``for $0<s <c \leq 1$, ``given a $c$-satisfiable $P$-CSP instance
$\Phi$, find an $s$-satisfying assignment for $\Phi$''. 
Max-CSPs and Cover-CSPs, as observed by Dinur and Kol~\cite{DinurK2013}, are very
different problems. For instance, if $P$ is an odd predicate, i.e, if
for every assignment $x$, either $x$ or its negation $x+\overline{1}$
satisfies $P$, then any $P$-CSP instance $\Phi$ has a trivial
2-covering   
any assignment and its negation. Thus, $\threelin$ and
$\threecnf$\footnote{$\mathsf{k\text{-}LIN}:\{0,1\}^k\to \{0,1\}$
	refers to the $k$-bit predicate defined by
	$\mathsf{k\text{-}LIN}(x_1,x_2,\dots,x_k) := x_1 \oplus x_2\oplus
	\cdots \oplus x_k$
	while $\mathsf{3\text{-}CNF}:\{0,1\}^3\to\{0,1\}$ refers to the
	$3$-bit predicate defined by $\mathsf{3\text{-}CNF}(x_1,x_2,x_3) :=
	x_1 \vee x_2\vee x_3$}, being odd predicates, are easy to cover
though they are hard predicates in the Max-CSP setting. The main
result of Dinur and Kol is that the $\fourlin$ predicate which accepts odd parities, in contrast
to the above, is hard to cover: for every constant $t \geq 2$,
\cov-$\fourlin$-\csp$(2,t)$ is $\mathrm{NP}$-hard. In fact, their arguments show
that \cov-$\fourlin$-\csp$(2,\Omega(\log\log \log n))$ is
quasi-$\mathrm{NP}$-hard.

Having observed that
CSPs based on odd predicates  
are easy to cover, Dinur
and Kol proceeded to ask the question ``are all non-odd-predicate CSPs
hard to cover?'' 
In a partial answer to this question, they showed
that assuming a covering variant of the
Unique Games Conjecture, 
\cov-\ugc$(c)$, if a predicate $P$ is not odd and there is a balanced
pairwise independent distribution on its support, then for
all constants $k$, \cov-$P$-\csp$(2c,k)$ is $\mathrm{NP}$-hard.
(Here, $c$ is a fixed
constant that depends on the covering variant of the
Unique Games Conjecture \cov-\ugc$(c)$.)   
See \expref{Sec.}{sec:prelims} for the exact
definition of the covering variant of the Unique Games Conjecture.

\subsection{Our results}

Our first result states that assuming the same covering variant of
the  
Unique Games Conjecture,  
\noindent\cov-\ugc$(c)$,  
of Dinur and
Kol~\cite{DinurK2013}, one can in fact show the covering hardness of
\emph{all} non-odd predicates $P$ over \emph{any} constant-size
alphabet $[q]$. The notion of odd predicate can be extended to any
alphabet in the following natural way: a predicate $P \subseteq
\setq^k$ is odd if for all assignments $x \in \setq^k$, there exists
$a \in \setq$ such that the assignment $x + \overline a$ satisfies $P$.

\begin{theorem}[Covering hardness of non-odd predicates]\label{thm:ug-hard} 
	
	Assuming \cov-\ugc$(c)$, for any constant-size  alphabet $[q]$, any
	constant $k\in \nat$ and any \emph{non-odd} predicate $P \subseteq
	\setq^k$, for all constants $t \in \nat$, the \cov-$P$-\csp$(2cq,t)$
	problem is $\mathrm{NP}$-hard.
	
\end{theorem}
Since odd predicates $P \subseteq [q]^k$ are trivially coverable with
$q$ assignments, the above theorem, gives a \emph{full characterization of
	hard-to-cover predicates} over any constant-size alphabet (modulo the
covering variant of the Unique Games Conjecture): a predicate is hard
to cover iff it is not odd.

We then ask if we can prove similar covering hardness results under
more standard complexity assumptions (such as $\mathrm{NP} \neq \mathrm{P}$ or the
exponential-time hypothesis (ETH)). Though we are not able to prove
that every non-odd predicate is hard under these assumptions, we give
sufficient conditions on the predicate $P$ for the corresponding
approximate covering problem to be quasi-$\mathrm{NP}$-hard. Recall that
$\twoklin \subseteq \{0,1\}^{2k}$ is the predicate corresponding to
the set of odd parity strings in $\{0,1\}^{2k}$.

\begin{theorem}[$\mathrm{NP}$ hardness of Covering]
	
	\label{thm:np-hard}
	Let $k \geq 2$. Let $P \subseteq \twoklin$ be any $2k$-bit predicate
	such there exist distributions $\mathcal P_0 , \mathcal P_1$
	supported on $\{0,1\}^k$ with the
	following properties:
	
	\begin{enumerate}
		
		\item the marginals of $\mathcal P_0$ and $\mathcal
		P_1$ on all $k$ coordinates 
		are uniform,  
		
		\item every $a \in \supp(\mathcal P_0)$ has
		even parity and every $b \in \supp(\mathcal P_1)$ has odd parity
		and furthermore, both $a\diamond b, b \diamond a \in P$, where $a\diamond
		b$ denotes the $2k$-bit string formed by the concatenation of
		strings $a$ and $b$.
	\end{enumerate}
	
	Then for all $\epsilon>0, r\gg1$, there is a  reduction
	from 3SAT to  \cov-$P$-\csp\ mapping a 3SAT instance $\Psi$ on
	$n$ variables to
	a \cov-$P$-\csp\ instance $\Phi$ of size $n^{O(r)}2^{2^{O(r)}}$ in 
	time $n^{O(r)}2^{2^{O(r)}}$ such that 
	\begin{itemize}
		\item {YES Case:} If the 3SAT formula $\Psi$ is satisfiable then there
		are $2$ assignments each satisfying $1-\epsilon$ of the constraints of
		$\Phi$, that together cover the instance $\Phi$. 
		\item {NO Case:} if the 3SAT formula $\Psi$ is not satisfiable, then
		the resulting instance $\Phi$ is not $\Omega_k(r) - O_k(\log(1/\epsilon))$
		coverable, even when considered as an instance of the (potentially
		larger) predicate $\twoklin$. 
	\end{itemize}
	
	In particular, unless $\mathrm{NP}$ $\subseteq \mathrm{DTIME}(2^{\poly \log n})$,  \cov-$P$-\csp$(2,\Omega(\log\log n))$ does not have a
	polynomial-time algorithm.
	
	If we assume $\mathrm{P}$ $\neq$ $\mathrm{NP}$ then
	\cov-$P$-\csp$(2,C)$ does not have a polynomial-time algorithm for
	any constant   
	$C > 2$.   
\end{theorem}

The furthermore clause in the soundness guarantee is in fact a
strengthening for the following reason: if two predicates $P,
Q$ satisfy $P \subseteq Q$ and $\Phi$ is a $c$-coverable $P$-CSP
instance, then the $Q$-CSP instance $\Phi_{P\to Q}$ obtained by taking
the constraint graph of $\Phi$ and replacing each $P$ constraint with
the weaker $Q$ constraint, is also $c$-coverable.

The following is a simple corollary of the above theorem.

\begin{corollary}
	Let $k\geq 2$ be even, $x,y \in \{0,1\}^k$ be distinct strings
	having even and odd parity, respectively, and
	$\overline x,\overline{y}$ denote the complements of $x$ and $y$,
	respectively.  For any predicate $P$ satisfying
	$$\twoklin \supseteq P \supseteq \{ x\diamond y,\ x \diamond \overline y,\
	\overline x \diamond y,\ \overline x \diamond \overline y, \ y\diamond x,\ y
	\diamond \overline x,\ \overline y \diamond x,\ \overline y \diamond
	\overline x \},$$
	unless $\mathrm{NP}$ $\subseteq \mathrm{DTIME}(2^{\poly\log n})$, the problem
	\cov-$P$-\csp$(2,\Omega(\log \log n))$ is not solvable in polynomial
	time.
\end{corollary}

This corollary implies the covering hardness of $\fourlin$ predicate
proved by Dinur and Kol ~\cite{DinurK2013} by setting $x:= 00$ and
$y:=01$. With respect to the covering
hardness of $\fourlin$, we note that we can considerably simplify the
proof of Dinur and Kol and in fact obtain a even stronger soundness
guarantee (see Theorem below). The stronger soundness guarantee in the
theorem below states that there are no large ($\geq 1/\poly\log n$
fractional-size) independent sets in the constraint graph and hence,
even the $4\text{-}\nae$-CSP instance\footnote{The $k\text{-}\nae$
	predicate over $k$ bits is given by $k\text{-}\nae=\{0,1\}^k
	\setminus \{\overline{0},\overline{1}\}$.} with the same constraint
graph as the given instance is not coverable using $\Omega(\log \log
n)$ assignments. Both the Dinur--Kol result and the above corollary
only guarantee (in the soundness case) that the $\fourlin$-CSP
instance is not coverable.

\begin{theorem}[Hardness of Covering $\fourlin$]
	
	\label{thm:4lin}
	Assuming that $\mathrm{NP}$ $\not\subseteq \mathrm{DTIME}(2^{\poly\log n})$, for all $
	\epsilon \in (0,1)$, there does not exist a polynomial-time
	algorithm that can distinguish between $\fourlin$-\csp\ instances of
	the following two types:
	
	\begin{itemize}
		
		\item YES Case : There are $2$ assignments such that each of them covers $1-
		\epsilon$ fraction of the constraints, and they
		together cover the entire instance.
		
		\item NO Case :  The largest independent set in the
		constraint graph of the instance is of fractional
		size at most $1/\poly\log n$. 
	\end{itemize}
\end{theorem}

\subsection{Techniques}

As one would expect, our proofs are very much inspired from the
corresponding proofs in Dinur and Kol~\cite{DinurK2013}. One of the
main complications in the proof of Dinur and Kol~\cite{DinurK2013} (as
also in the earlier work of Guruswami, H{\aa}stad and
Sudan~\cite{GuruswamiHS2002}) was the one of handling several
assignments simultaneously while proving the soundness analysis. For
this purpose, both these works considered the rejection probability
that all the assignments violated the constraint. This resulted in a
very tedious expression for the rejection probability, which made the
rest of the proof fairly involved. Holmerin~\cite{Holmerin2002} observed that
this can be considerably simplified if one instead proved a stronger soundness
guarantee that the largest independent set in the constraint graph is
small (this might not always be doable, but in the cases when it is,
it simplifies the analysis). We list below the further improvements in
the proof that yield our 
\expref{Theorems}{thm:ug-hard}, \ref{thm:np-hard} \expref{and}{thm:4lin}.

\paragraph{Covering-UG hardness for non-odd predicates (\expref{Theorem}{thm:ug-hard}).} 
Having observed that it suffices to prove an independent set analysis,
we observed that only very mild conditions on the predicate are
required to prove covering hardness. In particular, while Dinur and
Kol used the Austrin--Mossel test~\cite{AustrinM2009} which required
pairwise independence, we are able to import the long-code test of
Bansal and Khot~\cite{BansalK2010} which requires only 1-wise
independence. We remark that the Bansal--Khot Test was designed for a
specific predicate (hardness of finding independent sets in almost
$k$-partite $k$-uniform hypergraphs) and had imperfect
completeness. Our improvement comes from observing that their test
requires only 1-wise independence and furthermore that their
completeness condition, though imperfect, can be adapted to give a
2-cover composed of 2 nearly satisfying assignments using the \emph{duplicate label technique} of Dinur--Kol.  This enlarges
the class of non-odd predicates for which one can prove covering
hardness (see \expref{Theorem}{thm:basic-ug-hard}). We then perform a
sequence of reductions from this class of CSP instances to
CSP instances over all non-odd predicates to obtain the final
result. Interestingly, one of the open problems mentioned in the work
of Dinur and Kol~\cite{DinurK2013} was to devise ``direct'' reductions
between covering problems. The reductions we employ, strictly
speaking, are not ``direct'' reductions between covering problems,
since they rely on a stronger soundness guarantee for the source
instance (namely, large covering number even for the $\nae$ instance
on the same constraint graph), which we are able to prove in \expref{Theorem}{thm:basic-ug-hard}.

We give an overview of the dictatorship test gadget which when composed with a covering-UG instance, gives the required covering hardness result. Let $P \subseteq \setq^k$ be a predicate such that there exists $a \in\nae  $ and 
$$\nae \supset P \supseteq \{a+\bar b \mid b \in \setq \},$$
\ie,  $P$ accepts all shifts of a particular assignment $a\in [q]^k$
where $a\in \nae$. We are given a function $f : [q]^{2L}\rightarrow [q]$ and are interested in a $k$-query test, querying at $(x_1, x_2, \ldots, x_k)$ according to some distribution $\calD$, which has the following three properties:
\begin{enumerate}
	\item The accepting criteria of the test is $(f(x_1), f(x_2), \ldots, f(x_k))\in P$
	\item For every $i\in [L]$, the test should accept with probability $1$ if $f$ is either the $i$-th dictator or the $(i+L)$-th dictator.
	\item If $f$ is far from any dictator then the test, even with the predicate $P$ replaced by $\nae$, should reject with significant probability.
\end{enumerate}
We can think of the queries as a $k\times 2L$ matrix $X$ where the rows represent $x_1, x_2, \ldots, x_k$. Here is a distribution $\calD$ for which the test has all the above three properties: It will be a $L$-wise product distribution $\mu^{\otimes L}$, where $\mu$ is a distribution on $([q]^k)^2$ sampled uniformly from the set $S$, 
$$S := \left\{(y, y') \in \setq^k \times\setq^k \ |\  y\in \{a + 
\bar b\mid b\in \setq \} \vee y' \in \{a + \bar b\mid b\in \setq
\}\right\}.$$
For each $i\in [L]$ we sample the $i$-th and $(i+L)$-th columns of
$X$ independently from $\mu$. This completes the description of the distribution $\calD$. It is clear from the construction that the test with accepting criteria (1) satisfies (2) as either the $i$-th column or the $(i+L)$-th column contains an accepting assignment. The argument that (3) also holds for this test crucially depends on the properties of the distribution $\mu$ -- that each query $x_i$ is distributed uniformly in $\{0,1\}^{2L}$ and the distribution $\mu$ is \emph{connected} (see \expref{Definition}{def:connectedness}), when viewed as a probability space $(([q]^2)^k, \mu)$. Using both these properties of the distribution $\mu$, we can then apply the invariance principle to argue that the constrained (hyper)graph formed by the test distribution has a small independent set, which in turns imply (3).

\paragraph{Quasi-$\mathrm{NP}$ hardness result (\expref{Theorem}{thm:np-hard}).}
In this setting, we unfortunately are not able to use the
simplification arising from using the independent set analysis and
have to deal with the issue of several assignments. One of the steps
in the $\fourlin$ proof of Dinur and Kol (as in several others results in this area) involves
showing that a expression of the form $\E_{(X,Y)}\left[F(X)
F(Y)\right]$ is not too negative where $(X,Y)$ is not necessarily a
product distribution but the marginals on the $X$ and $Y$ parts are
identical. Observe that if $(X,Y)$ was a product distribution, then
the above expressions reduces to $\left(
\E_X\left[F(X)\right]\right)^2$, a positive quantity. Thus, the
steps in the proof involve constructing a tailor-made distribution
$(X,Y)$ such that the error in going from the correlated probability
space $(X,Y)$ to the product distribution $(X\otimes Y)$ is not too
much. More precisely, the quantity $$\left| \E_{(X,Y)}\left[F(X)
F(Y)\right] - \E_X\left[F(X)\right] \E_Y\left[F(Y)\right] \right
|,$$ is small. Dinur and Kol used a distribution tailor-made for the
$\fourlin$ predicate and used an invariance principle for correlated
spaces to bound the error while transforming it to a product
distribution. Our improvement comes from observing that one could use
an alternate invariance principle (see \expref{Theorem}{thm:inv-prin})
that works with milder restrictions and hence works for a wider class
of predicates. This invariance principle for correlated spaces
(\expref{Theorem}{thm:inv-prin}) is an adaptation of invariance
principles proved by Wenner~\cite{Wenner2013} and Guruswami and
Lee~\cite{GuruswamiL2018} in similar contexts. The rest of the proof
is similar to the $\fourlin$ covering hardness proof of Dinur and
Kol. 

\paragraph{Covering hardness of $\fourlin$
	(\expref{Theorem}{thm:4lin}).}
The simplified proof of the covering hardness of $\fourlin$ follows
directly from the above observation of using an independent set
analysis instead of working with several assignments. In fact, this
alternate proof eliminates the need for using results about correlated
spaces~\cite{Mossel2010}, which was crucial in the Dinur--Kol
setting. We further note that the quantitative improvement in the covering
hardness ($\Omega(\log\log n)$ over $\Omega(\log\log\log n)$) comes
from using a \LC\ instance with a better smoothness property (see
\expref{Theorem}{thm:lc-hard}).

\subsection{Recent work on approximate coloring}\label{sec:recentcolors}

Besides the work on covering complexity, the works most related to our
paper are the series of works that study the approximate coloring
complexity question, stated in the beginning of the introduction. Saket~\cite{Saket2014} showed that unless
$\mathrm{NP} \subseteq \mathrm{DTIME}(2^{\poly \log n})$, it is not
possible to color a 2-colorable 4-uniform hypergraph with $\poly \log
n$ colors. We remark that recently, with the discovery of the short
code~\cite{BarakGHMRS2015}, there has been a sequence of
works~\cite{DinurG2015,GuruswamiHHSV2017,KhotS2017,Varma2015,Huang2015} which
have considerably improved the status of the approximate coloring
question. In particular,
we know that it is quasi-$\mathrm{NP}$-hard to color a 2-colorable 8-uniform
hypergraph with $2^{(\log n)^c}$ colors for some constant $c \in
(0,1)$. Stated in terms of covering number, this result states that it
is quasi-$\mathrm{NP}$-hard to cover a 1-coverable $8\text{-}\nae$-CSP instance
with $(\log n)^c$ assignments. 
It is to be
noted that these results 
pertain to the covering complexity of specific predicates (such as
$\nae$) whereas our results are concerned with classifying which
predicates are hard to cover. It would be interesting if
\expref{Theorems}{thm:np-hard} \expref{and}{thm:4lin}
can be
improved to obtain similar hardness results (\ie, $\poly \log n$ as
opposed to $\poly\log \log n$). The main bottleneck here seems to be
reducing the uniformity parameter (namely, from 8).

\subsection*{Organization}

The rest of the paper is organized as follows. We start with some
preliminaries of \LC, covering CSPs and Fourier analysis in
\expref{Sec.}{sec:prelims}.
\expref{Theorems}{thm:ug-hard}, \ref{thm:np-hard}
\expref{and}{thm:4lin} are proved in \expref{Sections}{sec:ug-hard},
\ref{sec:np-hard} \expref{and}{sec:4lin}, respectively.


\section{Preliminaries}\label{sec:prelims}

\subsection{Covering CSPs}

We will denote the set 
$\{0,1,\cdots q-1\}$ by $\setq$. For $a\in \setq, \bar a \in [q]^k$ is the element 
with $a$ in all the $k$ coordinates (where $k$ and $q$ will be
implicit from the context).

\begin{definition}[$P$-\csp]
	
	\label{def:csp}
	For a predicate $P \subseteq \setq^k$, an instance of $P$-CSP
	is given by a (hyper)graph $G=(V,E)$, referred to as the {\em
		constraint graph}, and a literals function
	$L:E \rightarrow \setq^k$, where $V$ is a set of variables and
	$E \subseteq V^k$ is a set of constraints. An assignment
	$f:V \rightarrow \setq$ is said to \emph{cover} a constraint
	$e=(v_1,\cdots, v_k) \in E$, if
	$(f(v_1),\cdots,f(v_k))+L(e) \in P$, where addition is
	coordinate-wise modulo $q$. A set of assignments
	$F=\{ f_1,\cdots,f_c\}$ is said to \emph{cover} $(G,L)$, if
	for every $e\in E$, there is some $f_i \in F$ that covers $e$
	and $F$ is said to be a $c$\emph{-covering} for $G$. $G$ is
	said to be $c$\emph{-coverable} if there is a $c$-covering for
	$G$. If $L$ is not specified then it is the constant function
	which maps $E$ to $\bar 0$.
\end{definition}

\begin{definition}[\cov-$P$-\csp$(c,s)$]
	
	\label{def:cov-csp}
	For $P \subseteq \setq^k$ and $c,s \in \nat$, the \cov-$P$-\csp$(c,s)$ problem is, 
	given a $c$-coverable instance  $(G=(V,E),L)$ of $P$-\csp, find an 
	$s$-covering.
\end{definition}

\begin{definition}[Odd]
	
	\label{def:odd}
	A predicate $P \subseteq \setq^k$ is \emph{odd} if $\forall x \in \setq^k,
	\exists a \in \setq, x + \bar a \in P$, where addition is coordinate-wise
	modulo $q$.
\end{definition}

For odd predicates the covering problem is \emph{trivially solvable}, since any CSP 
instance on such a predicate is $q$-coverable by the $q$
translates of any assignment, \ie, $\{x + \bar a \mid a \in \setq\}$ is a
$q$-covering for any assignment $x \in \setq^k$.

\subsection{Label Cover}
\begin{definition}[\LC]
	
	\label{def:label-cover} An instance $G=(U,V,E,L,
	R,\{\pi_e\}_{e\in E})$ of the {\LC} constraint satisfaction
	problem consists of a bi-regular bipartite graph $(U,V,E)$,
	two sets of alphabets $L$ and $R$ and a projection map $\pi_e : R \rightarrow
	L$ for every edge $e\in E$. 
	Given a labeling $\ell : U \rightarrow L, \ell:V \rightarrow
	R$, an edge $e = (u,v)$ is said to be satisfied by $\ell$ if
	$\pi_e(\ell(v)) = \ell(u)$. 
	
	$G$ is said to be \emph{at most $\delta$-satisfiable} if every
	labeling satisfies at most a $\delta$ fraction of the
	edges. $G$ is said to be \emph{$c$-coverable} if there exist
	$c$ labelings such that for every vertex $u\in U$, one of the
	labelings satisfies all the edges incident on $u$.
	
	An instance of \UG\ is a label cover instance where $L=R$ and the constraints 
	$\pi$ are permutations.
\end{definition}

The hardness of \LC\ stated below follows from
the PCP Theorem \cite{AroraS1998,AroraLMSS1998}, Raz's Parallel
Repetition Theorem \cite{Raz1998} and a structural property proved by
H\aa stad \cite[Lemma 6.9]{Hastad2001}.
\begin{theorem}[Hardness of \LC]
	
	\label{thm:lc-hard} For every $r \in \N$, there is a
	deterministic $n^{O(r)}$-time reduction from a \TSAT\ instance
	of size $n$ to an instance $G=(U,V,E,[L],[R], \{\pi_e\}_{e
		\in E})$ of \LC\ with the following properties:
	
	\begin{enumerate}
		
		\item $|U|,|V| \leq n^{O(r)}$; $L,R \leq 2^{O(r)}$; $G$ is bi-regular with degrees bounded by $2^{O(r)}$.

		\item There exists a constant $c_0 \in (0,1/3)$ such that for any $v \in V$ and $\alpha
		\subseteq [R]$, for a random neighbor $u$,
		$$\E_u \left[ |\pi_{uv}(\alpha)|^{-1} \right] \leq |\alpha|^{-2c_0}, $$
		where $\pi_{uv}(\alpha) := \{ i\in [L] \mid \exists j\in \alpha \mbox{ s.t. } \pi_{uv}(j) = i\}$.
		This implies that
		$$\forall v, \alpha, \qquad Pr_u \left[ |\pi_{uv}(\alpha)| <|\alpha|^{c_0}\right] \leq \frac{1}{|\alpha|^{c_0}}.$$
		\item There is a constant $d_0 \in (0,1)$ such that,
		
		\begin{itemize}
			
			\item YES Case : If the \TSAT\ instance is
			satisfiable, then $G$ is 1-coverable.
			
			\item NO Case : If the \TSAT\ instance is
			unsatisfiable, then $G$ is
			at most $2^{-d_0r}$-satisfiable.
		\end{itemize}
	\end{enumerate}
\end{theorem}

Our characterization of hardness of covering CSPs is based on the
following conjecture due to Dinur and Kol~\cite{DinurK2013}. 

\begin{conjecture}[\cov-\ugc$(c)$]
	There exists $c\in \mathbb{N}$ such that for every sufficiently small $\delta>0$ 
	there exists $L\in \mathbb{N}$ such that the following holds. 
	Given an instance $G=(U,V,E,[L],[L], \{\pi_e\}_{e
		\in E})$ of \UG\, it is $\mathrm{NP}$-hard to distinguish between the following two cases:
	
	\begin{itemize}
		\item YES case: There exist $c$ assignments such that for 
		every vertex $u\in U$, at least one of the assignments satisfies all the edges touching u.
		
		\item NO case: Every assignment satisfies at most $\delta$ fraction of the
		edge constraints.
	\end{itemize}
\end{conjecture}

\subsection{Analysis of Boolean functions over probability spaces} 
For a function $f:\{0,1\}^L \rightarrow \R$, the \emph{Fourier
	decomposition} of $f$ is given by 
$$f(x) = \sum_{\alpha \in \{0,1\}^L}
\widehat f(\alpha) \chi_\alpha(x) \text{ where } \chi_\alpha(x) :=
(-1)^{\sum_{i=1}^L\alpha_i\cdot x_i} \text{ and }\widehat f(\alpha) :=
\E_{x \in \{0,1\}^L} f(x)\chi_\alpha(x).$$
We will use $\alpha$, also to denote the subset of $[L]$ for which it
is the characteristic vector. The \emph{Efron--Stein decomposition} is
a generalization of the Fourier decomposition to product distributions
of arbitrary probability spaces. Let $(\Omega, \mu)$ be a probability space and
$(\Omega^L,\mu^{\otimes L})$ be the corresponding product space. For a function
$f:\Omega^L \rightarrow \R$, the Efron--Stein decomposition of $f$ with
respect to the product space is given by
$$ f(x_1,\cdots, x_L) = \sum_{\beta \subseteq [L]} f_\beta(x),$$
where $f_\beta$ depends only on $x_i$ for $i\in \beta$ and for all
$\beta' \not\supseteq \beta , a \in \Omega^{\beta'}$, $\E_{x \in
	\mu^{\otimes R}} \left[ f_\beta(x) \mid x_{\beta'} = a \right]=0$. We
will be dealing with functions of the form $f:\{0,1\}^{dL}\rightarrow
\R$ for $d \in \N$ and $d$-to-$1$ functions $\pi:[dL] \rightarrow
[L]$. We will also think of such functions as $f: \prod_{i \in
	L}\Omega_i \rightarrow \R$ where $\Omega_i = \{0,1\}^d$ consists of
the $d$ coordinates $j$ such that $\pi(j) = i$. An Efron--Stein
decomposition of $f: \prod_{i \in L}\Omega_i \rightarrow \R$ over the
uniform distribution over $\{0,1\}^{dL}$, can be obtained from the
Fourier decomposition as 
\begin{equation}
	\label{eqn:fourier-efron}
	f_\beta(x) = \sum_{\alpha \subseteq [dL] : \pi(\alpha)=\beta} \widehat f(\alpha) \chi_\alpha.
\end{equation}
Let $\|f\|_2 :=
\E_{x\in \mu^{\otimes L}}[f(x)^2]^{1/2}$ and $\|f\|_\infty :=
\max_{x\in \Omega^{\otimes L}}|f(x)|$ . For $i \in [L]$, the influence of
the $i$-th
coordinate on $f$ is defined as follows.
$$\Inf_i[f] := \E_{x_1,\cdots, x_{i-1},x_{i+1},\cdots , x_L}\Var_{x_i}[f
(x_1,\cdots, x_L)]  = \sum_{\beta: i\in \beta} \|f_\beta\|^2_2.$$
For an integer $d$, the degree $d$ influence is defined as
$$\Inf_i^{\leq d}[f] := \sum_{\beta: i\in \beta, |\beta| \leq d} \|f_\beta\|^2_2.$$
It is easy to see that for Boolean functions, the sum of all the degree $d$ influences is at most $d$.

\begin{definition}
	\label{def:connectedness}
	Let $(\Omega^k, \mu)$ be a probability space. Let $S = \{ x\in \Omega^k 
	\mid \mu(x) > 0\}$. We say that  $S \subseteq \Omega^k$ is 
	\emph{connected} if for every $x, y\in S$, there is a sequence of strings 
	starting with $x$ and ending with $y$ such that every element in the sequence is
	in $S$ and every two adjacent elements differ in exactly one coordinate. 
\end{definition}

Let $\mu^{\otimes n}$ denote the $n$-wise product distribution of $\mu$. 

\begin{theorem}[{\cite[Proposition 6.4]{Mossel2010}}]
	
	\label{thm:invariance}
	Let $(\Omega^k, \mu)$ be a probability space such that the support of the 
	distribution $\supp(\mu) \subseteq \Omega^k$ is connected and the minimum probability of
	every atom in $\supp(\mu)$ is at least $\alpha$ for some
	$\alpha \in (0, \nicefrac{1}{2}]$. Furthermore, assume that the marginal of $\mu$ on each of the $k$ coordinates is uniform in $\Omega$. Then there exist
	continuous functions $\overline{\Gamma} : (0,1)\rightarrow (0,1)$ and 
	$\underline{\Gamma} : (0,1)\rightarrow (0,1)$ such that the following holds: 
	For every $\epsilon>0$, there exists $\tau > 0$ and an integer $d$ such that 
	if a function $f : \Omega^L \rightarrow [0,1]$ satisfies
	$$\forall i\in [L],  \Inf_i^{\leq d} (f) \leq \tau $$
	then 
	$$\underline{\Gamma}\left(\E_{(x_1,\ldots, x_k) \sim \mu^{\otimes L}}[f(x_1)]\right) -\epsilon \leq \E_{(x_1,\ldots, x_k) \sim \mu^{\otimes L}}\left[
	\prod_{j=1}^k f(x_j)\right] \leq \overline{\Gamma}\left(\E_{(x_1,\ldots, x_k) \sim \mu^{\otimes L}}[f(x_1)]\right) + \epsilon.$$
	There exists an absolute constant $C$ such that one can take $\tau = \displaystyle{\epsilon^
		{C\frac{\log(\nicefrac{1}{\alpha})\log(\nicefrac{1}{\epsilon})}{\epsilon
				\alpha^2}}}$ and $d = \log(\nicefrac{1}{\tau})\log(\nicefrac{1}{\alpha})$.
\end{theorem}

The following invariance principle for correlated spaces proved in 
\expref{Section}{sec:inv-prin} is an
adaptation of similar invariance principles (c.f.,
\cite[Theorem~3.12]{Wenner2013}, \cite[Lemma~B.3]{GuruswamiL2018}) to
our setting. 
\begin{theorem}[Invariance Principle for correlated spaces]
	\label{thm:inv-prin} Let $(\Omega_1^k \times \Omega_2^k, \mu)$ be a
	correlated probability space such that the marginal of $\mu$ on any
	pair of coordinates one each from $\Omega_1$ and $\Omega_2$ is a
	product distribution. Let $\mu_1 ,\mu_2$ be the marginals of $\mu$
	on $\Omega_1^k$ and $\Omega_2^k$, respectively. Let $X, Y$ be two
	random $k\times L$ dimensional matrices chosen as
	follows.  Independently  
	for every $i \in [L]$, the pair of columns $(x^i,y^i)
	\in \Omega_1^k \times \Omega_2^k$ is chosen from $\mu$. Let
	$x_i,y_i$ denote the $i$-th rows of $X$ and $Y$, respectively.  If
	$F: \Omega_1^L \rightarrow [-1,+1]$ and $G: \Omega_2^L \rightarrow
	[-1,+1]$ are functions such that
	$$\tau:= \sqrt{\sum_{i \in [L]}\Inf_i[F]\cdot \Inf_i[G]}  ~\text{ and } ~
	\Gamma := \max \left\{ \sqrt{\sum_{i \in [L]}\Inf_i[F]} \ , \sqrt{\sum_{i \in 
			[L]}\Inf_i[G]} \right\} \ ,$$ then
	\begin{equation}
		\label{eqn:inv-eqn}
		\abs{ \E_{(X,Y) \in \mu^{\otimes L}} \left[\prod_{i\in
				[k]}F(x_i) \cdot G(y_i)
			\right] - \E_{X \in \mu_1^{\otimes L}} \left[\prod_{i\in [k]}F(x_i)\right]
			\cdot \E_{Y \in \mu_2^{\otimes L}} \left[\prod_{i\in [k]}G(y_i)\right] } \leq 
		2^{O(k)} \Gamma \tau\ .
	\end{equation}
\end{theorem}


\section{Covering-UG Hardness of Covering CSPs}\label{sec:ug-hard}

In this section, we prove the following theorem, which in turn implies
\expref{Theorem}{thm:ug-hard} (see below for proof).

\begin{theorem} 
	
	\label{thm:basic-ug-hard}
	Let $[q]$ be any constant-size alphabet and $k \geq
	2$. Recall that $\nae := \setq^k \setminus \{\bar b\mid b \in
	\setq \}$. Let $P \subseteq \setq^k$ be a predicate such that
	there exists $a \in\nae  $ and $\nae \supset P \supseteq \{a
	+\bar b \mid b \in \setq \}$. 
	Assuming \cov-\ugc$(c)$, for every sufficiently small constant $\delta>0$ it is $\mathrm{NP}$-hard to 
	distinguish between $P$-{\csp} instances $\calG=(\calV,\calE)$ of the 
	following two cases:
	
	\begin{itemize}
		
		\item YES Case : $\calG$ is $2c$-coverable.
		
		\item NO Case : $\calG$ does not have an independent set of fractional 
		size $\delta$.
	\end{itemize}
\end{theorem} 

\begin{proof}[Proof of {\expref{Theorem}{thm:ug-hard}}]
	Let $Q$ be an arbitrary
	non-odd predicate  
	i.e, $Q\subseteq \setq^k
	\setminus \{h+\bar b\mid b\in \setq\}$ for some $h \in \setq^k$. Consider
	the predicate $Q' \subseteq \setq^k$ defined as $Q' := Q -h := \{ g -
	h \mid g\in Q\}$, where the operation `$-$' refers to coordinate-wise
	substraction performed$\pmod q$. Observe
	that $Q' \subseteq \nae$. Given any $Q'$-CSP instance $\Phi$ with
	literals function $L(e) = \overline{0}$, consider the $Q$-CSP instance
	$\Phi_{Q'\to Q}$ with literals function $M$ given by $M(e) :=
	\overline{h}, \forall e$. It has the same constraint graph as
	$\Phi$. Clearly, $\Phi$ is $c$-coverable iff $\Phi_{Q'\to Q}$ is
	$c$-coverable. Thus, it suffices to prove the result for any predicate
	$Q' \subseteq \nae$ with literals function $L(e) = \overline
	0$\footnote{This observation~\cite{DinurK2013} that the cover-$Q$
		problem for any non-odd predicate $Q$ is equivalent to the
		cover-$Q'$ problem where $Q' \subseteq \nae$ shows the centrality of
		the $\nae$ predicate in understanding the covering complexity of any
		non-odd predicate.}. We will consider two cases, both of which will
	follow from \expref{Theorem}{thm:basic-ug-hard}.
	
	
	Suppose the predicate
	$Q'$ satisfies $Q' \supseteq \{a +\bar b \mid b \in \setq \}$ for some
	$a\in [q]^k$. Then this predicate $Q'$ satisfies the hypothesis of
	\expref{Theorem}{thm:basic-ug-hard} and the theorem follows if we show
	that the soundness guarantee of \expref{Theorem}{thm:basic-ug-hard}
	implies that in \expref{Theorem}{thm:ug-hard}. Any instance in the NO case of
	\expref{Theorem}{thm:basic-ug-hard}, is not $t:=\log_q
	(1/\delta)$-coverable even on the $\nae$-\csp\ instance with the same
	constraint graph. This is because any $t$-covering for the
	$\nae$-\csp\ instance gives a coloring of the constraint graph using
	$q^t$ colors, by choosing the color of every variable to be a string
	of length $t$ and having the corresponding assignments in each
	position in $[t]$. Hence the $Q'$-\csp\ instance is also not
	$t$-coverable
	.

	Suppose $Q' \not\supseteq \{a +\bar b \mid b \in \setq \}$ for all $a \in
	\setq^k$. Then consider the predicate $P = \{ a +\bar b \mid a \in Q', b
	\in \setq\} \subseteq \nae$. Notice that $P$ satisfies the conditions
	of \expref{Theorem}{thm:basic-ug-hard} and if the $P$-\csp\ instance is
	$t$-coverable then the $Q'$-\csp\ instance is $qt$-coverable. Hence
	a YES instance of \expref{Theorem}{thm:basic-ug-hard} maps to a $2cq$-coverable
	$Q$-\csp\ instance and NO instance maps to an instance with covering
	number at least $\log_q(1/\delta)$, where the latter follows from the fact that the covering number of the instance as a $Q'$-CSP is at least the covering number of it as a $P$-CSP.
\end{proof}

\paragraph{}

We now prove \expref{Theorem}{thm:basic-ug-hard} by giving a reduction
from an instance $G=(U,V,E,[L], [L],$ $\{\pi_e\}_{e\in E})$ of {\UG} as
in \expref{Definition}{def:label-cover}, to an instance $\mathcal G =
(\mathcal V, \mathcal E)$ of a $P$-{\csp} for any predicate $P$ that
satisfies the conditions mentioned. As stated in the introduction, we
adapt the long-code test of Bansal and Khot~\cite{BansalK2010} for
proving the hardness of finding independent sets in almost $k$-partite
$k$-uniform hypergraphs to our setting. The set of variables $\mathcal
V$ is $V\times \setq^{2L}$. Any assignment to $\mathcal V$ is given by
a set of functions $f_v:\setq^{2L} \rightarrow \setq$, for each $v\in
V$. The set of constraints $\mathcal E$ is given by the following test
which checks whether $f_v$'s are long codes of a good labeling to
$V$. There is a constraint corresponding to all the variables that are
queried together by the test.

\paragraph{Long Code Test $\calT_1$}

\begin{enumerate}
	
	\item Choose $u\in U$ uniformly and $k$ neighbors $w_1, \ldots , w_k \in V$ 
	of $u$ uniformly and independently at random. 
	
	\item Choose  a random matrix $X$ of dimension $k \times 2L$
	as follows. Let $X^i$ denote the $i$-th column of $X$. Independently for each $i\in [L]$, choose 
	$(X^i,X^{i+L}) $ uniformly at random from the set
	\begin{align}\label{eq:defS}
		S &:= \left\{(y, y') \in \setq^k \times\setq^k \ |\  y\in \{a + 
		\bar b\mid b\in \setq \} \vee y' \in \{a + \bar b\mid b\in \setq
		\}\right\}.
	\end{align}
	\item Let $x_1,\cdots,x_k$ be the rows of matrix $X$. Accept iff
	$$ \left(f_{w_1} (x_1\circ \pi_{uw_1}), f_{w_2} (x_2\circ \pi_{uw_2}), \cdots, f_{w_k}
	(x_k\circ \pi_{uw_k}) \right) \in P,$$
	where $x\circ \pi$ is the string defined as $(x \circ \pi)(i) := x_{\pi(i)}$
	for $i \in [L]$ and $(x\circ \pi)(i) := x_{\pi(i-L)+L}$ otherwise.
\end{enumerate}

Before plunging into the formal analysis of the reduction, let us see the intuition behind the test.
The test is designed so that if the functions $f_{w_1}, f_{w_2},
\ldots, f_{w_k}$ are dictator functions satisfying the UG-constraints
associated with the common neighbor $u$ or their $L$ \emph{shifts},
then the test passes. This is obvious as the bit pattern from the
locations queried is either $y$ or $y'$, one of which belongs to the
predicate $P$. This gives a 2-covering of the instance: one corresponds to the actual dictator functions satisfying the UG-constraints and another consists of $L$ shifts of those dictator functions. Another property of
the set $S$  
that   
is used in the test is that it defines a probability space
that   
is \emph{connected}. This will be used in the soundness analysis of the test. We now prove the completeness and the soundness of the reduction.

\begin{lemma}[Completeness]
	
	If  the \UG\ instance $G$ is $c$-coverable then the $P$-{\csp} instance 
	$\mathcal G$ is $2c$-coverable.
\end{lemma}

\begin{proof}
	Let $\ell_1, \ldots, \ell_c:U\cup V \rightarrow [L]$ be a $c$-covering for $G$
	as described in \expref{Definition}{def:label-cover}.
	We will show that the
	$2c$ assignments given by $f_v^i(x) := x_{\ell_i(v)},
	g_v^i(x):= x_{\ell_i(v)+L}, i = 1, \dots, c$  form a 
	$2c$-covering of $\mathcal G$. Consider any $u \in U$ and let $\ell_i$ be the
	labeling that covers all the edges incident on $u$. For any $(u,w_j)_{j\in 
		\{1,\cdots,k\}} \in E$ and $X$ chosen by the long code test
	$\mathcal T_1$, the vector $(f_{w_1}^i
	(x_1\circ \pi_{uw_1}),\cdots, f_{w_k}^i(x_k \circ
	\pi_{uw_k}))$ gives the 
	$\ell_i(u)$-th column of $X$. Similarly the above expression corresponding to 
	$g^i$ gives the $(\ell_i(u)+ L)$-th column of the matrix $X$. Since, for all 
	$i\in [L]$, either $i$-th column  or $(i+L)$-th column of $X$ contains element
	from $\{a +\bar{b} \mid b \in \setq\} \subseteq P$,  either $(f_{w_1}^i(x_1\circ \pi_{uw_1})
	,\cdots, f_{w_k}^i(x_k \circ \pi_{uw_k})) \in P$ or $(g_{w_1}^i(x_1\circ \pi_{u
		w_1}),\cdots, g_{w_k}^i(x_k \circ \pi_{uw_k})) \in P$. Hence the set of
	$2c$ assignments $\{f_v^i, g_v^i\}_{i\in \{1,\cdots,c\}}$ covers all constraints in $\calG$.
\end{proof}

To prove soundness, we show that the set $S$, as
defined in \expeqref{Equation}{eq:defS}, is connected, so
that \expref{Theorem}{thm:invariance} is applicable. 
For this, we view $S \subseteq [q]^k \times [q]^k$ as a subset of
$(\setq^2)^k$ as follows: the element $(y,y') \in S$ is mapped to the element
$((y_1,y'_1),\cdots,(y_k,y'_k)) \in ([q]^2)^k$. 
\begin{claim}
	\label{claim:rel_conn}
	Let $\Omega = \setq^2$. The set $S \subset \Omega^k$
	is connected.
\end{claim}

\begin{proof}
	Consider any $x:=(x^1,x^2),y := (y^1,y^2) \in S
	\subset [q]^{k} \times [q]^{k}$. Suppose both $x^1,y^1
	\in \{a + \bar b\mid b\in \setq \}$, then it is easy to
	come up with a sequence of strings belonging to $S$,
	starting with $x$ and ending with $y$ such that 
	consecutive strings differ in at most $1$
	coordinate,. Now
	suppose $x^1,y^2 \in \{a + \bar b\mid b\in \setq \}$.
	First we come up with a sequence from $x$ to $z:=(z^1, z^2)$ such
	that $z^1:=x^1$ and $z^2=y^2$, and then another
	sequence for $z$ to $y$.
\end{proof}

\begin{lemma}[Soundness]
	
	\label{lemma:ug_soundness}
	For every constant $\delta>0$, there exists a constant $s$ such that, if $G$ is 
	at most $s$-satisfiable then $\calG$ does not have an independent set of size 
	$\delta$.
\end{lemma}

\begin{proof}
	Let $I \subseteq \calV$ be an independent set of fractional
	size $\delta$ in the constraint graph. For every variable $v
	\in V$, let $f_v : \setq^{2L} \rightarrow \{0,1\}$ be the
	indicator function of the independent set restricted to the
	vertices that correspond to $v$. For a vertex $u\in U$, let
	$N(u) \subseteq V$ be the set of neighbors of $u$ and define $f_u(x)
	:= \E_{w\in N(u)} [f_w(x\circ \pi_{uw})] $. Since $I$ is an
	independent set, we have
	\begin{equation}
		\label{eqn:ug_frachyp}
		0 = \E_{u, w_i,\ldots,w_k} \E_{X \sim \calT_1}\left[ \prod_{i=1}^k f_{w_i}(x_i 
		\circ \pi_{uw_i}) \right]= \E_u\E_{X\sim \calT_1}\left[\prod_{i=1}^k f_u(x_i)
		\right] .
	\end{equation}
	Since the bipartite graph $(U,V,E)$ is left regular and $|I|\geq \delta|V|$, 
	we have $\E_{u,x}[f_u(x)] \geq \delta$. By an averaging argument, for at least 
	$\nicefrac{\delta}{2}$ fraction of the vertices $u\in U$, $\E_x[f_u(x)] \geq \nicefrac{\delta}{2}$.
	Call a vertex $u\in U$ \emph{good} if it satisfies this property. A string 
	$x\in \setq^{2L}$ can be thought as an element from $(\setq^2)^L$ by grouping 
	the pair of coordinates $x_i, x_{i+L}$. Let $\overline{x} \in (\setq^2)^L $ denotes 
	this grouping of $x$, \ie, $j$-th coordinate of $\overline{x}$ is $(x_j, x_{j+L})$ is distributed u.a.r. in $\setq^2$. With this grouping, the function $f_u$ can be viewed  
	as $f_u : (\setq^2)^L\rightarrow \{0,1\}$. From
	\expeqref{Equation}{eqn:ug_frachyp}, we have that for any $u\in U$,
	$$\E_{X\sim \calT_1}\left[\prod_{i=1}^k f_u(\overline{x}_i)\right] = 0. $$ 
	By \expref{Claim}{claim:rel_conn}, for all $j\in [L]$ the tuple $((\overline{x}_
	1)_j, \ldots, (\overline{x}_k)_j)$ (corresponding to columns $(X^j, $$X^{j+L})$ 
	of $X$) is sampled from a distribution whose support is a connected set. 
	Hence for a good vertex $u \in U$, we can apply \expref{Theorem}{thm:invariance}
	with $\epsilon = \underline{\Gamma}(\delta/2)/2$ to get that there exists $j 
	\in [L], d \in \N, \tau > 0$ such that $\Inf_j^{\leq d} (f_u) > \tau $. We will use this fact to give
	a randomized labeling for $G$. Labels for vertices $w \in V, u \in U$ will be chosen
	uniformly and independently from the sets
	$$\Lab(w) := \left\{ i\in [L] \mid \Inf_i^{\leq d} (f_w) \geq \frac{\tau}{2}
	\right\},\Lab(u) := \left\{ i\in [L] \mid \Inf_i^{\leq d} (f_u) \geq \tau
	\right\}.$$
	By the above argument
	(using \expref{Theorem}{thm:invariance}), we
	have that for a good vertex $u$,
	$\Lab(u)\neq \emptyset$. Furthermore, since the sum of degree $d$ influences is at most $d$, the above
	sets have size at most $ 2d/\tau$. Now, for any $j\in \Lab(u)$, we have
	\begin{align*}
		\tau &< \Inf_j^{\leq d} [f_u] = \sum_{S: j \in S, |S| \leq d} \|f_{u,S}\|^2 = \sum_{S: j \in S, |S| \leq d} \left\|\E_{w\in N(u)} \left[f_{w,\pi_{uw}^	{-1}(S)}\right]\right\|^2\quad(\text{By Definition.})\\
		& \leq \sum_{S: j \in S , |S| \leq d} \E_{w\in N(u)} \left\|f_{w,\pi_{uw}^
			{-1}(S)}\right\|^2\quad  = \E_{w\in N(u)} \Inf_{\pi_{uw}^{-1}(j)}^{\leq d}[ f_w ].\quad (\text{By Convexity of square.})
	\end{align*}
	Hence, by another averaging argument, there exists at least
	$\nicefrac{\tau}{2}$ fraction of neighbors $w$ of $u$ such that $
	\Inf_{\pi_{uw}^{-1}(j)}^{\leq d} ( f_w ) \geq \nicefrac{\tau}{2}$
	and hence $\pi^{-1}_{uw}(j) \in \Lab(w)$. Therefore, for a
	good vertex $u\in U$, at least
	$\nicefrac{\tau}{2}\cdot\nicefrac{\tau}{2d}$ fraction of
	edges incident on $u$ are satisfied in expectation. Also, at least
	$\nicefrac{\delta}{2}$ fraction of vertices in $U$ are good, it
	follows that the expected fraction of edges that are satisfied
	by this random labeling is at least
	$\nicefrac{\delta}{2}\cdot\nicefrac{\tau}{2}\cdot\nicefrac{\tau}{2d}$. Choosing
	$s < \nicefrac{\delta\tau^2}{8d}$ 
	completes the proof.
\end{proof}


\section{NP Hardness of Covering CSPs}\label{sec:np-hard}

In this section, we prove \expref{Theorem}{thm:np-hard}. We give a
reduction from an instance of a {\LC}, $G=(U,V,E,[L],
[R],\{\pi_e\}_{e\in E})$ as in \expref{Definition}{def:label-cover}, to
a $P$-CSP instance $\mathcal G = (\mathcal V, \mathcal E)$ for any
predicate $P$ that satisfies the conditions mentioned in
\expref{Theorem}{thm:np-hard}. The reduction and proof is similar to that of Dinur
and Kol~\cite{DinurK2013}. The main difference is that they used a
test and invariance principle very specific to the $\fourlin$
predicate, while we show that a similar analysis can be performed under
milder conditions on the test distribution.

We assume that $R=dL$ and $\forall i \in [L], e \in E,
|\pi_e^{-1}(i)|=d$. This is done just for simplifying the notation and
the proof does not depend upon it. The set of variables $\mathcal V$
is $V\times \{0,1\}^{2R}$. Any assignment to $\mathcal V$ is given by
a set of functions $f_v:\{0, 1\}^{2R} \rightarrow \{0,1\}$, for each
$v\in V$. The set of constraints $\mathcal E$ is given by the
following test,   
which checks whether $f_v$'s are long codes of a good
labeling to $V$.

\paragraph{Long Code test  $\mathcal T_2$}

\begin{enumerate}
	
	\item Choose $u\in U$ uniformly and $v,w \in V$ neighbors of
	$u$ uniformly and independently at random. For $i \in [L]$, define the sets $B_{uv}(i):=
	\pi_{uv}^{-1}(i), B'_{uv}(i):= R+\pi_{uv}^{-1}(i) $ and similarly for $w$.
	
	\item \label{step:2} Choose matrices $X,Y$ of dimension $k \times 2dL$ as
	follows. For $S\subseteq [2dL]$, we denote by $X|_{S}$ the
	submatrix of $X$ restricted to the columns $S$. Independently for each $i\in [L]$, choose $c_1 \in \{0,1\}$ uniformly and
	
	\begin{enumerate}
		
		\item if $c_1=0$, choose $\left(X|_{B_{uv}(i)\cup B'_{uv}(i)}, Y|_{B_{uw}(i)\cup B'_{uw}(i)}\right)$  from $\calP_0^{\otimes 2d}\otimes \calP_1^{\otimes 2d}$,
		\item if $c_1=1$, choose $\left(X|_{B_{uv}(i)\cup B'_{uv}(i)}, Y|_{B_{uw}(i)\cup B'_{uw}(i)}\right)$ from $\calP_1^{\otimes 2d}\otimes\calP_0^{\otimes 2d}$.
	\end{enumerate} 
	
	\item \label{step:3} Perturb $X, Y$ as follows. Independently for each 
	$i \in [L]$, choose $c_2 \in \{*,0,1\}$ as follows:
	$$\Pr[c_2 = *] = 1-2\epsilon, \mbox{ and } \Pr[c_2 = 1] = \Pr[c_2 = 0 ] = \epsilon.$$ 
	Perturb the $i$-th matrix block $\left(X|_{B_{uv}(i)\cup
		B'_{uv}(i)}, Y|_{B_{uw}(i)\cup B'_{uw}(i)}\right)$ as follows:
	\begin{enumerate}
		\item\label{step:30} if $c_2 = *$, leave the matrix
		block $\left(X|_{B_{uv}(i)\cup B'_{uv}(i)}, Y|_{B_{uw}(i)\cup B'_{uw}(i)}\right)$
		unperturbed,
		\item\label{step:3a} if $c_2=0$, choose $\left(X|_{ B'_{uv}(i)} ,Y|_{ B'_{uw}(i)}\right)$ uniformly from $  \{0,1\}^{k \times d} \times \{0,1\}^{k \times d}$,
		
		\item\label{step:3b} if $c_2=1$, choose $\left(X|_{B_{uv}(i)} ,Y|_{B_{uw}(i)}\right)$  uniformly from $\{0,1\}^{k \times d} \times \{0,1\}^{k \times d}$.
	\end{enumerate} 
	
	\item\label{step:4} Let $x_1,\cdots,x_k$ and $y_1,\cdots, y_k$ be t
	he rows
	of the matrices $X$ and $Y$, respectively. Accept if $$
	\left(f_v(x_1 ),\cdots, f_v(x_k ), f_w(y_1 ),\cdots,f_w(y_k)
	\right) \in P .$$
	
\end{enumerate}

\begin{lemma}[Completeness]
	
	If $G$ is an YES instance of {\LC}, then there exists $f,g$
	such that each of them covers $1-\epsilon$ fraction of
	$\mathcal E$ and they together cover all of $\mathcal E$.
	
\end{lemma}

\begin{proof}
	Let $\ell:U\cup V \rightarrow [L] \cup [R]$ be a labeling to
	$G$ that satisfies all the constraints. Consider the
	assignments $f_v(x) := x_{\ell(v)} $ and $g_v(x):=
	x_{R+\ell(v)}$ for each $v \in V$. First consider the
	assignment $f$. For any $(u,v),(u,w) \in E$
	and $x_1,\cdots,x_k,y_1,\cdots,y_k$ chosen by the long code
	test $\mathcal T_2$, $(f_v(x_1 ),\cdots, f_v(x_k )),$ $
	(f_w(y_1),\cdots,f_w(y_k))$ gives the $\ell(v)$-th and
	$\ell(w)$-th column of the matrices $X$ and $Y$, respectively.
	Since $\pi_{uv}(\ell(v)) = \pi_{uw}(\ell (w))$, they are jointly
	distributed either according to $\calP_0 \otimes \calP_1$ or
	$\calP_1 \otimes \calP_0$ after \expref{Step}{step:2}. The probability that
	these rows are perturbed in \expref{Step}{step:3b} is at most $\epsilon$.
	Hence with probability $1-\epsilon$ over the test
	distribution, $f$ is accepted. A similar argument shows
	that the test accepts $g$ with probability $1-\epsilon$.
	Note that in \expref{Step}{step:3}, the columns given by $f,g$, are never
	re-sampled uniformly together. Hence they together cover
	$\mathcal G$.
\end{proof}

Now we will show that if $G$ is a NO instance of \LC\ then no $t$
assignments can cover the $\twoklin$-\csp\ with constraint hypergraph
$\mathcal G$. For the rest of the analysis, we will use $+1,-1$
instead of the symbols $0,1$. Suppose for contradiction, there exist $t$ assignments $f_1,\cdots,f_t : \{ \pm 1 \}^{2R}
\rightarrow \{ \pm 1\}$ that form a $t$-cover to $\mathcal G$. The
probability that all the $t$ assignments are rejected in \expref{Step}{step:4} is
{\small
	\begin{equation}
		\label{eqn:arith-1}
		\E_{u,v,w} \E_{\mathcal T_2} \left[ \prod_{i=1}^t \frac{1}{2} \left(  
		\prod_{j=1}^kf_{i,v}(x_j)  f_{i,w}(y_j )  + 1\right)\right]= 
		\frac{1}{2^t} + \frac{1}{2^t}\sum_{\emptyset\subset S\subseteq \{1,\cdots, t\} }
		\E_{u,v,w} \E_{\mathcal T_2}\left[ \prod_{j=1}^kf_{S,v}(x_j)  f_{S,w}(y_j
		)\right].
	\end{equation}
}
where $f_{S,v}(x) := \prod_{i \in S}f_{i,v}(x)$. Since the $t$
assignments form a $t$-cover, the LHS in
\expeqref{Equation}{eqn:arith-1} is $0$ and hence, there exists an $S \neq
\emptyset$ such that
\begin{equation}
	\label{eq:nphard_ari}
	\E_{u,v,w} \E_{\mathcal T_2}\left[\prod_{j =1}^kf_{S,v}(x_j)  f_{S,w}(y_j)\right]\leq -1/(2^t-1).
\end{equation}
The following lemma shows that this is not possible if $t$ is not too
large, thus proving that there does not a exist $t$-cover.
\begin{lemma}[Soundness]
	
	\label{lemma:soundness_nphard} Let $c_0\in(0,1)$ be the
	constant from \expref{Theorem}{thm:lc-hard} and $S\subseteq \{1,\cdots,t\},
	|S|>0$. If $G$ is at most $s$-satisfiable then
	$$ \E_{u,v,w}~\E_{ X,Y \in \mathcal T_2} \left[\prod_{i =1}^k f_{S,v}(x_i )f_{S,w}(y_i )\right]  \geq -O(ks^{c_0/8}) -2^{O(k)}\frac{s^{(1-3c_0)/8}}{\epsilon^{3/2c_0}}.$$
\end{lemma}

\begin{proof}
	Notice that for a fixed $u$, the distribution of $X$ and $Y$
	have identical marginals. Hence the value of the above
	expectation, if calculated according to a distribution
	that   
	is the direct
	product of the marginals, is positive. We will first show that
	the expectation can change by at most $O(ks^{c_0/8})$ in moving to
	an \emph{attenuated} version of the functions 
	(see \expref{Claim}{claim:atten}). Then we will show that	 the
	error incurred by changing the distribution to the product
	distribution of the marginals has absolute value at most
	$\displaystyle{2^{O(k)}\frac{s^{(1-3c_0)/8}}{\epsilon^{3/2c_0}}}$ (see \expref{Claim}{claim:inv-app}). This is done
	by showing that there is a labeling to $G$ that satisfies an
	$s$ fraction of the constraints if the error is more than
	$\displaystyle{2^{O(k)}\frac{s^{(1-3c_0)/8}}{\epsilon^{3/2c_0}}}$. 
	
	For the rest of the analysis, we write $f_v$ and $f_w$ instead
	of $f_{S,v}$ and $f_{S,w}$, respectively. Let $f_v
	=\sum_{\alpha \subseteq [2R]} \widehat f_v(\alpha)\chi_\alpha$
	be the Fourier decomposition of the function and for $\gamma
	\in (0,1)$, let $T_{1-\gamma} f_v := \sum_{\alpha \subseteq
		[2R]} (1-\gamma)^{|\alpha|}\widehat f_v(\alpha)\chi_\alpha$.
	The following claim is similar to a lemma of Dinur and Kol~\cite[Lemma 4.11]{DinurK2013}. 
	The only difference in the proof is that, we use the \emph{smoothness} from Property 2 of 
	\expref{Theorem}{thm:lc-hard} 
	(which was shown by H\aa stad  \cite[Lemma 6.9]{Hastad2001}).
	\begin{claim}
		
		\label{claim:atten}Let $\gamma := s^{(c_0+1)/4}\epsilon^{1/c_0}$ where $c_0$
		is the constant from \expref{Theorem}{thm:lc-hard}.
		$$\left| \E_{u,v,w}\E_{\mathcal T_2} \left[\underbrace{\prod_{i=1}^k f_v(x_i)f_w(y_i
			)}_{\Delta_0}\right]  - \E_{u,v,w}\E_{\mathcal T_2} \left[\underbrace{\prod_{i=1}^k T_{1-\gamma}
			f_v(x_i) T_{1-\gamma}f_w(y_i)}_{\Delta_1}\right] \right|\leq O(ks^{c_0/8}).$$
		
	\end{claim}

	\begin{proof} The claim bounds the change in the expectation when we change the expression $\Delta_0$ to $\Delta_1$. The expression $\Delta_0$ is a product of $2k$ functions and $\Delta_1$ is the product of the same functions after applying the $T_{1-\gamma}$ operator to each of these functions. We prove the claim by bounding the error with $O(s^{c_0/8})$ when we add an extra $T_{1-\gamma}$ operator each time. Thus, the total error will be $O(ks^{c_0/8})$ by doing the telescoping sum and using the triangle inequality.
		
		For notational convenience, we bound the error when we add the first $T_{1-\gamma}$. The effect of adding all the remaining subsequent $T_{1-\gamma}$ operators can be analyzed in a similar way.
		\begin{equation}
			\label{eqn:add-noise}
			\left|\E_{u,v,w}\E_{\mathcal T_2} \left[\prod_{i=1}^k f_v(x_i)f_w(y_i)\right]  - \E_{u,v,w}\E_{\mathcal T_2} \left[\left(\prod_{i=1}^{k-1} f_v(x_i) f_w(y_i)\right) f_v(x_k) T_{1-\gamma}f_w(y_k) \right] \right|\leq O(s^{c_0/8}).
		\end{equation}
		Recall that $X,Y$ denote the matrices chosen by test $\calT_2$. Let
		$Y_{-k}$ be the matrix obtained from $Y$ by removing the $k$-th row
		and $F^{u,v,w}(X,Y_{-k}):= \left(\prod_{i = 1}^{k-1} f_v(x_i)
		f_w(y_i)\right) f_v(x_k)$. Then, 
		\expeqref{Eq.}{eqn:add-noise} 
		can be rewritten as
		\begin{equation}
			\label{eqn:add-noise2}\left|\E_{u,v,w}\E_{\mathcal T_2}
			\left[F^{u,v,w}(X,Y_{-k})\left(I - T_{1-\gamma}
			\right)f_w(y_k)\right]  \right|\leq O(s^{c_0/8}).
		\end{equation}
		Let $U$ be the operator that maps functions on the variable $y_k$, to one on the variables
		$(X,Y_{-k})$ defined by
		$$(Uf)(X,Y_{-k}) := \E_{y_k | X,Y_{-k}} f(y_k).$$
		Let
		$G^{u,v,w}(X,Y_{-k}) :=
		\left(U(I-T_{1-\gamma})f_w\right)(X,Y_{-k})$.
		Note that $\E_{(X,Y)\sim \calT_2}\ G^{u,v,w}(X,Y_{-k}) = 0$. This is because $\E_{(X,Y)\sim \calT_2}\ G^{u,v,w}(X,Y_{-k}) = \E_{y_k \sim \{0,1\}^{2L}} ((I-T_{1-\gamma})f_w)(y_k) = \widehat{((I-T_{1-\gamma})f_w)}(\emptyset)$, where the marginal distribution on $y_k$ is uniform in $\{0,1\}^{2L}$.
		Finally, by construction,  
		$\E_{(X,Y)\sim \calT_2}\ G^{u,v,w}(X,Y_{-k}) = 0$ follows,
		since $f_w$ is an odd function.  
		The domain
		of $G^{u,v,w}$ can be thought of as $(\{0,1\}^{2k-1})^{2dL}$ and the test
		distribution on any row is independent across the blocks
		$\{B_{uv}(i)\cup B'_{uv}(i)\}_{i \in [L]}$. We now think of $G^{u,v,w}$ as
		having domain $\prod_{i \in [L]} \Omega_i$ where
		$\Omega_i= (\{0,1\}^{2k-1})^{2d}$ corresponds to the set of rows in
		$B_{uv}(i)\cup B'_{uv}(i)$. Let the following be the Efron--Stein
		decomposition of $G^{u,v,w}$ with respect to $\calT_2$,
		$$G^{u,v,w}(X,Y_{-k}) = \sum_{\alpha \subseteq [L]} G^{u,v,w}_\alpha(X,Y_{-k}).$$
		The following technical claim follows from a result similar to \cite[Lemma 4.7]{DinurK2013} and then using \cite[Proposition 2.12]{Mossel2010}. We defer its proof to \expref{Section}{apx:tech-claim}. Here we use the role of the random variable $c_2$ in $\mathcal T_2$, which helps to break the perfect correlation between one row and rest of the rows restricted to the columns $B_{uv}(i)\cup B'_{uv}(i)$ for all $i\in [L]$.
		
		\begin{claim}
			\label{claim:np_normbound}
			For $\alpha \subseteq[L]$
			\begin{equation}\label{eqn:tech-claim}
				\left\| G^{u,v,w}_\alpha \right\|^2 \leq (1-\epsilon)^{|\alpha|}\sum_{\beta \subseteq [2R]: \widetilde{\pi}_{uw}(\beta) = \alpha }\left(1-(1-\gamma)^{2|\beta|}\right)\widehat f_w(\beta)^2
			\end{equation}
			where   $\widetilde{\pi}_{uw}(\beta) := \{ i \in
			[L]: \exists j \in [R], (j \in \beta \vee j+R \in \beta)
			\wedge \pi_{uv}(j)=i\}$. 
			
		\end{claim}
		Substituting the Efron--Stein decomposition of $G^{u,v,w},F^{u,v,w}$ into the LHS of 
		\expeqref{Eq.}{eqn:add-noise}  
		gives 
		\begin{align*}
			\label{eqn:diff}
			\left|\E_{u,v,w}\E_{\mathcal T_2}
			\left[F^{u,v,w}(X,Y_{-k})\left(I - T_{1-\gamma}
			\right)f_w(y_k)\right]  \right| &= \left|\E_{u,v,w}   \E_{\calT_2}   F^{u,v,w}(X,Y_{-k})  G^{u,v,w}(X,Y_{-k}) \right|\\
			\substack{\text{(by orthonormality of}\\\text{ Efron--Stein decomposition)}}~~~&=  \left| \E_{u,v,w}\sum_{\alpha \subseteq [L]} \E_{\calT_2}    F^{u,v,w}_\alpha (X,Y_{-k}) G^{u,v,w}_\alpha(X,Y_{-k}) \right|\\
			\text{(by Cauchy--Schwarz inequality)}~~~&\leq  \E_{u,v,w}\sqrt{\sum_{\alpha \subseteq [L]}   \|F^{u,v,w}_\alpha\|^2} \cdot \sqrt{\sum_{\alpha \subseteq [L]}  \|G^{u,v,w}_\alpha\|^2}  \\
			\text{(Using $\sum_{\alpha \subseteq [L]}\| F^{u,v,w}_\alpha \|^2 = \|F^{u,v,w}\|_2^2 =1$)}~~~&\leq  \E_{u,w}\sqrt{\sum_{\alpha \subseteq [L]}  \|G^{u,v,w}_\alpha\|^2} .
		\end{align*}
		Using concavity of square root and substituting for $\|G^{u,v,w}_\alpha\|^2$ from \expeqref{Equation}{eqn:tech-claim}, we get that the above is
		not greater than  
		\begin{align*}
			\sqrt{  \E_{ u,w} \sum_{\alpha \subseteq [L]}\sum_{\substack{\beta \subseteq [2R]:\\ \widetilde{\pi}_{uw}(\beta) = \alpha }} \underbrace{(1-\epsilon)^{|\alpha|}\left(1-(1-\gamma)^{2|\beta|}\right)\widehat f_w(\beta)^2}_{=:\term_{u,w}(\alpha,\beta)  } }.
		\end{align*}
		
		We will now break the above summation into three different parts and bound each part separately.
		\begin{align*}
			\Theta_0 &:= \E_{u,w} \sum_{\substack{\alpha,\beta : |\alpha|\geq \frac{1}{\epsilon s^{c_0/4}} }} \term_{u,w}(\alpha,\beta), & 
			\Theta_1 := \E_{u,w} \sum_{\substack{\alpha,\beta : |\alpha|< \frac{1}{\epsilon s^{c_0/4}}\\ |\beta| \leq \frac{2}{s^{1/4}{\epsilon}^{1/c_0} }}} \term_{u,w}(\alpha,\beta),\\
			\Theta_2 &:= \E_{u,w} \sum_{\substack{\alpha,\beta : |\alpha|< \frac{1}{\epsilon s^{c_0/4}}\\ |\beta| > \frac{2}{s^{1/4}{\epsilon}^{1/c_0} }}} \term_{u,w}(\alpha,\beta).&
		\end{align*}
		
		\paragraph{Upper bound for $\Theta_0$.} 
		When $\displaystyle{|\alpha| > \frac{1}{\epsilon s^{c_0/4}}}$, $(1-\epsilon)^{|\alpha|} < s^{c_0/4}$.
		Also since $f_w$ is $\{+1,-1\}$ valued, sum of squares of Fourier
		coefficient is $1$. Hence $|\Theta_0|< s^{c_0/4}$.
		
		\paragraph{Upper bound for $\Theta_1$.} 
		When $\displaystyle{|\beta| \leq \frac{2}{s^{1/4}{\epsilon}^{1/c_0}}}$, 
		$$1-(1-\gamma)^{2|\beta|}\leq 1-\left(1-\frac{4}{s^{1/4}{\epsilon}^{1/c_0}}\gamma\right) = \frac{4}{s^{1/4}{\epsilon}^{1/c_0}}\gamma = 4s^{c_0/4}.$$
		Again since the sum of squares of Fourier coefficients is $1$, $|\Theta_1| \leq 4s^{c_0/4}$.
		
		\paragraph{Upper bound for $\Theta_2$.}

		From Property 2 of \expref{Theorem}{thm:lc-hard}, we have that for any 
		$v\in V$ and $\beta$ with $\displaystyle{|\beta|> \frac{2}{s^{1/4}{\epsilon}^{1/c_0}}}$, the probability 
		that $|\widetilde{\pi}_{uv}(\beta)| < 1/\epsilon s^{c_0/4}$, for a random neighbor $u$, is at most 			$\epsilon s^{c_0/4}$. Hence $|\Theta_2|\leq s^{c_0/4}$.
		
	\end{proof}

	Fix $u,v,w$ chosen by the
	test. Recall that we thought of
	$f_v$ as having domain $\prod_{i \in [L]} \Omega_i$ where
	$\Omega_i= \{0,1\}^{2d}$ corresponds to the set of coordinates
	in $B_{uv}(i)\cup B'_{uv}(i)$. Since the grouping of
	coordinates depends on $u$, we define $\overline{\Inf}^u_i[f_v]:= \Inf_{i}[f_v] $ where $i\in [L]$ for explicitness. From 
	\expeqref{Equation}{eqn:fourier-efron}, 
	$$\overline{\Inf}^u_i[f_v] =
	\sum_{\alpha \subseteq [2dL]: i \in \widetilde{\pi}_{uv}(\alpha) }
	\widehat f_v (\alpha)^2,$$ where $\widetilde{\pi}_{uv}(\alpha) := \{ i \in
	[L]: \exists j \in [R], (j \in \alpha \vee j+R \in \alpha)
	\wedge \pi_{uv}(j)=i\}$.
	
	\begin{claim}
		
		\label{claim:inv-app}
		Let $\tau_{u,v,w}:= \sum_{i \in [L]} \overline{\Inf}^u_{i}[T_{1-\gamma}
		f_v]\cdot \overline{\Inf}^u_{i}[T_{1-\gamma}f_w]$.
		\begin{align*}
			\E_{u,v,w}\left|\E_{\mathcal T_2} \left[\prod_{i=1 }^k T_{1-\gamma}f_v
			(x_i) T_{1-\gamma}f_w(y_i)\right] - \E_{\mathcal T_2} \left[ \prod_{i=1}^k T_{1-\gamma}f_v(x_i) \right] \E_{\mathcal T_2} \left[ 				\prod_{i=1}^k T_{1-\gamma}f_w(y_i) \right] \right| \\
			\leq 2^{O(k)}\sqrt{\frac{\E_{u,v,w}\tau_{u,v,w}}{\gamma}}.
		\end{align*}
	\end{claim}
	
	\begin{proof}
		It is easy to check that $\sum_{i \in[L]}
		\overline{\Inf}^u_{i}[T_{1-\gamma}f_v] \leq 1/\gamma$
		(c.f., \cite[Lemma~1.13]{Wenner2013}).
		For any $u,v,w$, since the test distribution satisfies the conditions of \expref{Theorem}{thm:inv-prin}, we get 
		$$	\left|\E_{\mathcal T_2} \left[\prod_{i=1}^k T_{1-\gamma}f_v
		(x_i) T_{1-\gamma}f_w(y_i)\right] - \E_{\mathcal T_2} \left[ \prod_{i=1}^k T_{1-\gamma}f_v(x_i) \right] \E_{\mathcal T_2} \left[ \prod_{i=1}^k
		T_{1-\gamma}f_w(y_i) \right] \right| \leq 2^{O(k)}\sqrt{\frac{\tau_{u,v,w}}{\gamma}}
		.$$
		The claim follows by taking expectation over $u,v,w$ and using the concavity of square root.
	\end{proof}
	
	{}From \expref{Claims}{claim:inv-app} \expref{and}{claim:atten} 
	and using the fact the the marginals of the test distribution 
	$\calT_2$ on $(x_1,\ldots, 
	x_k)$ is the same as marginals on $(y_1,\ldots, y_k)$, for $\gamma := s^{(c_0+1)/4}\epsilon^{1/c_0}$, we get
	{\small
		\begin{equation}
			\label{eq:inv_finalbound}
			\E_{u,v,w}~\E_{ X,Y \in \mathcal T_2} ~\left[\prod_{i=1}^k f_v(x_i )f_w(y_i) \right] 
			\geq - O(ks^{c_0/8})  -2^{O(k)}\sqrt{\frac{\E_{u,v,w}\tau_{u,v,w}}{\gamma}} + \E_{u} 
			\left( \E_v\E_{\mathcal T_2} \left[ \prod_{i=1}^k T_{1-\gamma}f_v(x_i) 
			\right]\right)^2.
		\end{equation}
	}
	
	If $\tau_{u,v,w}$ in expectation is large, there is a standard way of decoding the assignments to a
	labeling to the label cover instance, as shown in \expref{Claim}{claim:decode}.
	
	\begin{claim}
		\label{claim:decode}
		If $G$ is an at most $s$-satisfiable instance of {\LC} then
		$$\E_{u,v,w} \tau_{u,v,w} \leq \frac{s}{\gamma^2}.$$ 
	\end{claim}
	
	\begin{proof}
		Note that $\sum_{\alpha \subseteq [2R]}(1-
		\gamma)^{|\alpha|}\widehat{f}_v(\alpha)^2 \leq 1$. We
		will give a randomized labeling to the \LC\ instance.
		For each $v\in V$, choose a random $\alpha \subseteq [2R]$
		with probability $(1-
		\gamma)^{|\alpha|}\widehat{f}_v(\alpha)^2$ and assign
		a uniformly random label $j$ in $\alpha$ to $v$; if
		the label $j\geq R$, change the label to $j-R$ and
		with the remaining probability assign an arbitrary
		label. For $u \in U$, choose a random neighbor $w \in
		V$ and a random $\beta \subseteq [2R]$ with
		probability $(1-\gamma)^{|\beta|}
		\widehat{f}_w(\beta)^2$, choose a random label $\ell$
		in $\beta$ and assign the label $\widetilde{\pi}_{uw}(\ell)$ to
		$u$. With the remaining probability, assign an
		arbitrary label. The fraction of edges satisfied by
		this labeling is at least 
		$$ \E_{u,v,w} \sum_{i \in
			[L]}~~ \sum_{\substack{(\alpha,\beta): i \in
				\widetilde{\pi}_{uv}(\alpha), i \in \widetilde{\pi}_{uw}(\beta) }}
		\frac{(1-\gamma)^{|\alpha|+|\beta|}}{|\alpha|\cdot
			|\beta|} \widehat f_v(\alpha)^2 \widehat
		f_w(\beta)^2.$$ 
		Using the fact that $1/r \geq
		\gamma(1-\gamma)^{r}$ for every $r> 0$ and $\gamma \in
		[0,1]$, we lower bound $\nicefrac{1}{|\alpha|}$ and
		$\nicefrac{1}{|\beta|}$ by $\gamma(1-\gamma)^{|\alpha|}$
		and $\gamma(1-\gamma)^{|\beta|}$, respectively. The
		above is then
		not less than   
		$$\gamma^2 \E_{u,v,w}
		\sum_{i \in [L]} \left(\sum_{\substack{\alpha: i \in
				\widetilde{\pi}_{uv}(\alpha) }} (1-\gamma)^{2|\alpha|} \widehat
		f_v(\alpha)^2\right)\left( \sum_{\substack{\beta: i
				\in \widetilde{\pi}_{uw}(\beta) }} (1-\gamma)^{2|\beta|} \widehat
		f_w(\beta)^2 \right)= \gamma^2 \E_{u,v,w}
		\tau_{u,v,w}.$$ Since $G$ is at most $s$-satisfiable, the
		labeling can satisfy at most
		an  
		$s$ fraction of constraints and the
		right-hand side of the above equation is at most $s$. 
	\end{proof}
	\expref{Lemma}{lemma:soundness_nphard} follows from the above claim 
	and \expeqref{Equation}{eq:inv_finalbound}.
\end{proof}

\begin{proof}[Proof of {\expref{Theorem}{thm:np-hard}}.]
	Using \expref{Theorem}{thm:lc-hard}, the size of the CSP instance $\calG$ produced 
	by the reduction is $N = n^{r}2^{2^{O(r)}}$ and the parameter $s\leq 2^{-d_0 r}$
	. Setting $r = \Theta(\log \log n)$, gives that $N=2^{\poly \log n}$ for a
	constant $k$. \expref{Lemma}{lemma:soundness_nphard} and 
	\expeqref{Equation}{eq:nphard_ari} imply that
	$$ O(ks^{c_0/8}) +2^{O(k)}\frac{s^{(1-3c_0)/8}}{\epsilon^{3/2c_0}}\geq 
	\frac{1}{2^t -1}.$$
	Since $k$ is a constant, this gives that $t = \Omega(\log \log n)$.
	
	For every constant
	$C > 2$,   
	by choosing $r$ a large enough constant, we get the hardness result assuming $\mathrm{P}$ $\neq$ $\mathrm{NP}$.
\end{proof}

\subsection{Proof of \texorpdfstring{\expref{Claim}{claim:np_normbound}}{Claim 4.4}}\label{apx:tech-claim}


We will be reusing the notation introduced in the long code test $\calT_2$. We denote the $k \times 2d$ dimensional matrix $X|_{B(i)\cup B'(i)}$ by $X^i$ and $Y|_{B(i)\cup B'(i)}$ by $Y^i$. Also by $X^i_j$, we mean the $j$-th row of the matrix $X^i$ and $Y^i_{-k}$ is the first $k-1$ rows of $Y^i$. The spaces of the random variables $X^i,X^i_j,Y^i_{-k}$ will be denoted by $\calX^i,\calX^i_j, \calY^i_{-k}$.

Before we proceed to the proof of claim, we need a few definitions and
lemmas related to correlated spaces defined by Mossel~\cite{Mossel2010}.
\begin{definition}
	\label{def:correlation}
	Let $(\Omega_1 \times \Omega_2, \mu)$ be a finite correlated space, the correlation between $\Omega_1$ and $\Omega_2$ with respect to $\mu$ us defined as 
	$$\rho(\Omega_1, \Omega_2; \mu) := \mathop{\max}_{\substack{f : \Omega_1 \rightarrow \R, \E[f]  = 0 , \E[f^2]\leq 1 \\ g: \Omega_2 \rightarrow \R, \E[g] = 0 , \E[g^2]\leq 1} }  \E_{(x,y) \sim \mu }[ |f(x)g(y)|] .$$
\end{definition}

\begin{definition}[Markov Operator]
	\label{def:markovop}
	Let $(\Omega_1 \times \Omega_2, \mu)$ be a finite correlated space, the Markov operator, associated with this space, denoted by $U$, maps a function $g : \Omega_2 \rightarrow \R$ to functions $Ug : \Omega_1 \rightarrow \R$ by the following map:
	$$ (Ug) (x) := \E_{(X,Y)\sim \mu}[g(Y) \mid X=x ].$$
\end{definition}

\noindent The following results (due to Mossel~\cite{Mossel2010}) provide a way
to give an upper bound on the correlation of correlated spaces.  
\begin{lemma}[{\cite[Lemma~2.8]{Mossel2010}}]
	\label{lemma:mossel_corr}
	Let $(\Omega_1\times \Omega_2, \mu)$ be a finite correlated space. Let $g : \Omega_2 \rightarrow \R$ be such that $\E_{(x,y)\sim \mu}[g(y)] = 0$ and $\E_{(x,y)\sim \mu}[g(y)^2]\leq 1$. Then, among all functions $f : \Omega_1\rightarrow \R$ that satisfy $\E_{(x,y)\sim \mu}[f(x)^2]\leq 1$, the maximum value of $|\E[f(x)g(y)]|$ is given as:
	$$\left|\E[f(x)g(y)]\right| = \sqrt{\E_{(x,y)\sim \mu}[(Ug(x))^2]} .$$
\end{lemma}
\begin{proposition}[{\cite[Proposition~2.11]{Mossel2010}}]
	\label{prop:mossel_prop211}
	Let  $(\prod_{i=1}^n\Omega_i^{(1)} \times \prod_{i=1}^n\Omega_i^{(2)}, \prod_{i=1}^n\mu_i)$ be a product correlated space. Let $g :\prod_{i=1}^n\Omega_i^{(2)} \rightarrow \mathbb{R}$ be a function and $U$ be the Markov operator
	mapping functions 
	from the 
	space $\prod_{i=1}^n\Omega_i^{(2)}$ to 
	functions on space $\prod_{i=1}^n\Omega_i^{(1)}$. If $g = \sum_{S\subseteq [n]}g_S$ and $Ug = \sum_{S\subseteq [n]} (Ug)_S$ be the Efron--Stein
	decompositions of $g$ and $Ug$, respectively, then,
	$$ (Ug)_S = U(g_S)$$
	\ie, the Efron--Stein decomposition commutes with Markov operators.
\end{proposition}

\begin{proposition}[{\cite[Proposition~2.12]{Mossel2010}}]
	\label{prop:mossel_prop212}
	Assume the setting of \expref{Proposition}{prop:mossel_prop211} and
	furthermore assume that $\rho(\Omega_i^{(1)}, \Omega_i^{(2)}; \mu_i)
	\leq \rho$ for all $i\in [n]$. Then for all $g$ it holds that
	$$ \|U(g_S) \|_2 \leq \rho^{|S|}\|g_S\|_2.$$
\end{proposition}

\noindent We will prove the following claim.

\begin{claim}
	\label{claim:corr_np_def}
	For each $i \in [L]$, 
	$$ \rho\left(\calX^i\times\calY^i_{-k} , \calY_k^i ; \calT_2^i\right) \leq \sqrt{1-\epsilon}. $$
\end{claim}

\noindent Before proving this claim, first let's see how it leads to the proof of \expref{Claim}{claim:np_normbound}.

\begin{proof}[Proof of {\expref{Claim}{claim:np_normbound}}]
	\expref{Proposition}{prop:mossel_prop211} shows that the  Markov operator $U$ commutes with taking the Efron--Stein decomposition. Hence, $G^{u,v,w}_\alpha := (U((I - T_{1-\gamma})f_w))_\alpha =U((I - T_{1-\gamma})(f_w)_\alpha),$ where $(f_w)_\alpha$ is the Efron--Stein decomposition of $f_w$ w.r.t. the marginal distribution of $\calT_2$ on $\prod_{i=1}^L\calY_{k}^{i}$, 
	which is a uniform distribution. Therefore, $(f_w)_\alpha = \sum_{\substack{\beta \subseteq [2R],\\ \widetilde{\pi}_{uw}(\beta) = \alpha }}\hat{f_w}(\beta)\chi_\beta$.
	Using \expref{Proposition}{prop:mossel_prop212} and \expref{Claim}{claim:corr_np_def}, we have
	\begin{align*}
		\|G^{u,v,w}_\alpha\|_2^2 = \|U ((I - T_{1-\gamma})(f_w)_\alpha)\|_2^2 &\leq (\sqrt{1-\epsilon})^{2|\alpha|}\|(I - T_{1-\gamma})(f_w)_\alpha\|_2^2\\
		& = (1-\epsilon)^{|\alpha|} \sum_{\beta \subseteq [2R] : \widetilde{\pi}_{uw}(\beta) = \alpha} \left(1-(1-\gamma)^{2|\beta|}\right)\hat{f}_w(\beta)^2,
	\end{align*}
	where the norms are with respect to the marginals of $\calT_2$ in the corresponding spaces.
\end{proof}

\begin{proof}[Proof of {\expref{Claim}{claim:corr_np_def}}]
	Recall the random variable $c_2 \in \{*,0,1\}$ defined in
	\expref{Step}{step:3} of test $\calT_2$ . Let $g$ and $f$ be the
	functions that satisfies $\E[g] = \E[f]=0$ and $\E[g^2], \E[f^2] \leq
	1$ such that $\rho\left(\calX^i\times\calY^i_{-k} , \calY_k^i ;
	\calT_2^i\right) = \E[|fg|]$. Define the \emph{Markov Operator}
	$$ Ug(X^i, Y^i_{-k}) = \E_{(\tilde{X}, \tilde{Y}) \sim \calT_2^i}[ g(\tilde{Y}_k) \mid { (\tilde{X}, \tilde{Y}_{-k})} = (X^i, Y^i_{-k}) ]. $$
	By \expref{Lemma}{lemma:mossel_corr}, we have
	\begin{align*}
		\rho\left(\calX^i\times\calY^i_{-k} , \calY_k^i ; \calT_2^i\right)^2 &\leq  \E_{\calT_2^i} [Ug(X^i, Y^i_{-k})^2]\\
		& \hspace{-25pt}= (1-2\epsilon)\E_{\calT_2^i} [Ug(X^i, Y^i_{-k})^2  \mid  c_2 = *] + \epsilon\E_{\calT_2^i} [Ug(X^i, Y^i_{-k})^2  \mid  c_2 = 0] + \\
		&\hspace{5pt}\epsilon\E_{\calT_2^i} [Ug(X^i, Y^i_{-k})^2  \mid  c_2 = 1]\\
		& \hspace{-25pt}\leq (1-2\epsilon) + \epsilon\E_{\calT_2^i} [Ug(X^i, Y^i_{-k})^2  \mid  c_2 = 0] + \epsilon\E_{\calT_2^i} [Ug(X^i, Y^i_{-k})^2  \mid  c_2 = 1],
	\end{align*}
	where the last inequality uses the fact that $\E_{\calT_2^i} [Ug(X^i, Y^i_{-k})^2  \mid  c_2 = *] = \E[g^2]$,  
	which is at most $1$. Consider the case when $c_2=0$. By definition, we have
	\begin{align*}
		\hspace{-8pt}\E_{\calT_2^i} [Ug(X^i, Y^i_{-k})^2  \mid  c_2 = 0] &= \hspace{-5pt}\E_{\left(\substack{X^i, \\Y^i_{-k}}\right) \sim \calT_2^i} \left(\E_{ (\tilde{X}, \tilde{Y}) \sim \calT_2^i}[ g(\tilde{Y}_k) \mid {(\tilde{X}, \tilde{Y}_{-k})} = (X^i, Y^i_{-k}) \wedge c_2 = 0]\hspace{-3pt}\right)^2\hspace{-4pt}.
	\end{align*}
	Under the conditioning, for any fixed value of $X^i, Y^i_{-k}$, the value of $\tilde{Y}_k|_{B'(i)}$ is a uniformly random string whereas $\tilde{Y}_k|_{B(i)}$ is a fixed string (since the $parity$ of all columns in $B(i)$ is $1$).  
	Let $\calU$ be the uniform distribution on $\{-1,+1\}^d$ and $\mathcal{P}(X^i, Y^i_{-k}) \in \{ +1,-1\}^{d}$ denotes the column wise parities of 
	$\left[\substack{X^i|_{B(i)}\\ Y^i_{-k}|_{B(i)}}\right]$.

	\begin{align*}
		\hspace{-10pt}\E_{\calT_2^i} [Ug(X^i, Y^i_{-k})^2  \mid  c_2 = 0] &=\E_{X^i, Y^i_{-k} \sim \calT_2^i}\left(\E_{ (\tilde{X}, \tilde{Y}) \sim \calT_2^i}\left[ g(\tilde{Y}_k) \mid \substack{{(\tilde{X}, \tilde{Y}_{-k})} = (X^i, Y^i_{-k}) \wedge\\ c_2 = 0}\right] \right)^2\\
		&= \E_{\substack{X^i, Y^i_{-k} \sim \calT_2^i, \\z=\mathcal{P}(X^i, Y^i_{-k})}}\left(\E_{ r \sim \calU}[ g(-z,r) ]\right)^2\\
		&= \E_{z \sim \calU}\left(\E_{ r \sim \calU}[ g(z,r) ]\right)^2 ~~~ \text{(since marginal on $z$ is uniform)}\\  
		&=  \E_{z \sim \calU} \left(\E_{r\in \calU} \sum_{\alpha\subseteq B(i) \cup B'(i)} \hat{g}(\alpha) \chi_\alpha(z, r)\right)^2\\
		&=  \E_{z \sim \calU}\left(\sum_{\alpha\subseteq B(i) \cup B'(i)} \hat{g}(\alpha) \E_{r\in \calU}[\chi_\alpha(z, r)]\right)^2\\
		&=  \E_{z \sim \calU}\left(\sum_{\alpha\subseteq B(i)}\hat{g}(\alpha) \chi_\alpha(z)\right)^2 
		\qquad\qquad = \quad  
		\sum_{\alpha\subseteq B(i)}\hat{g}(\alpha)^2. 
	\end{align*}
	Similarly we have,
	\begin{align*}
		\E_{\calT_2^i} [Ug(X^i, Y^i_{-k})^2  \mid  c_2 = 1]  &= \sum_{\alpha\subseteq B'(i)}\hat{g}(\alpha)^2.
	\end{align*}
	Now we can bound the correlation as follows. 
	\begin{align*}
		\rho\left(\calX^i\times\calY^i_{-k} , \calY_k^i ; \calT_2^i\right)^2 \leq &(1-2\epsilon) + \epsilon\sum_{\alpha\subseteq B(i)}\hat{g}(\alpha)^2 + \epsilon\sum_{ \alpha\subseteq B'(i)}\hat{g}(\alpha)^2 \\
		\leq& (1-2\epsilon) + \epsilon\sum_{\alpha\subseteq B(i)\cup B'(i)}\hat{g}(\alpha)^2 ~~~\text{(Using $\hat{g}(\phi) = \E[g]=0$)}\\
		\leq& (1-\epsilon). ~~~~~~~~~~~\text{(Using $\E[g^2]\leq 1$ and Parseval's Identity)}
	\end{align*}
\end{proof}


\section{Improvement to covering hardness of \texorpdfstring{$\fourlin$}.}\label{sec:4lin}

In this section, we prove \expref{Theorem}{thm:4lin}. We give a
reduction from an instance of {\LC}, $G=(U,V,E,[L],
[R],\{\pi_e\}_{e\in E})$ as in \expref{Definition}{def:label-cover}, to
a $\fourlin$-{\csp} instance $\mathcal G = (\mathcal V, \mathcal
E)$. The set of variables $\mathcal V$ is $V\times \{0,1\}^{2R}$. Any
assignment to $\mathcal V$ is given by a set of functions $f_v:\{0,
1\}^{2R} \rightarrow \{0,1\}$, for each $v\in V$.  The set of
constraints $\mathcal E$ is given by the following
test,  
which checks
whether $f_v$'s are long codes of a good labeling to $V$.

\noindent
\paragraph{Long Code test $\calT_3$}

\begin{enumerate}
	
	\item Choose $u\in U$ uniformly and neighbors $v,w \in V$ of
	$u$ uniformly and independently at random.
	
	\item Choose $x,x',z,z'$ uniformly and independently from
	$\{0,1\}^{2R}$ and $y$ from $\{0,1\}^{2L}$. Choose
	$(\eta,\eta') \in \{0,1\}^{2L} \times \{0,1 \}^{2L}$ as
	follows.  
	Independently for each $i \in [L]$,
	set $(\eta_i,\eta_{L+i}, \eta'_i,\eta'_{L+i})$ to    
	\begin{enumerate}
		\item $(0,0,0,0)$ with probability $1-2\epsilon$,
		\item $(1,0,1,0)$ with probability $\epsilon$ and
		\item $(0,1,0,1)$ with probability $\epsilon$.
	\end{enumerate}
	
	\item For $y \in \{0,1\}^{2L}$, let $y \circ \pi_{uv} \in
	\{0,1\}^{2R}$ be the string such that $(y \circ \pi_{uv})_i
	:= y_{\pi_{uv}(i)}$ for $i \in [R]$ and $(y \circ
	\pi_{uv})_i := y_{\pi_{uv}(i-R)+L}$ otherwise. Given
	$\eta\in \{0,1\}^{2L}, z \in \{0,1\}^{2R}$, the string
	$\eta\circ \pi_{uv} \cdot z \in \{0,1\}^{2R}$ is 
	obtained by taking coordinate-wise product of $\eta \circ
	\pi_{uv}$ and $z$. Accept
	iff
	\begin{equation}
		\label{eqn:test3}
		f_v(x)+f_v(x+y\circ \pi_{uv}+\eta\circ \pi_{uv} \cdot z)+f_w(x')+f_w(x'+y\circ\pi_{uw}
		+\eta'\circ \pi_{uw}\cdot z' + \overline{1})  = 1\pmod 2.
	\end{equation}
	(Here by addition of strings, we mean the
	coordinate-wise sum modulo 2.)
	
\end{enumerate}

\begin{lemma}[Completeness]
	
	If $G$ is an YES instance of {\LC}, then there exists $f,g$
	such that each of them covers $1-\epsilon$ fraction of
	$\mathcal E$ and they together cover all of $\mathcal E$.
	
\end{lemma}

\begin{proof}
	Let $\ell:U\cup V \rightarrow [L] \cup [R]$ be a labeling to
	$G$ that satisfies all the constraints. Consider the
	assignments given by $f_v(x) := x_{\ell(v)} $ and $g_v(x):=
	x_{R+\ell(v)}$ for each $v \in V$. On input $f_v$, for any
	pair of edges $(u,v),(u,w) \in E$, and $x,x',z,z',\eta ,\eta',y$ chosen by the
	long code test $\mathcal T_3$, the LHS in \expeqref{Eq.}{eqn:test3} evaluates to
	$$x_{\ell(v)} + x_{\ell(v)}+y_{\ell(u)} +
	\eta_{\ell(u)}z_{\ell(v)} + x'_{\ell(w)}
	+x'_{\ell(w)}+y_{\ell(u)} + \eta'_{\ell(u)}z'_{\ell(w)}+1 = \eta_{\ell(u)}z_{\ell(v)}+ \eta'_{\ell(u)}z'_{\ell(w)}+1.$$
	Similarly for $g_v$, the expression evaluates to $\eta_{L+\ell(u)}z_{R+\ell(v)} +
	\eta'_{L+\ell(u)} z'_{R+\ell(w)}+1. $ Since
	$(\eta_{i},\eta'_{i})=(0,0)$ with probability $1-\epsilon$,
	each of $f,g$ covers $1-\epsilon$ fraction of $\calE$. Also
	for $i \in [L]$ whenever $(\eta_{i},\eta'_{i})=(1,1)$,
	$(\eta_{L+i},\eta'_{L+i})=(0,0)$ and vice versa. So one of the
	two evaluations above is $1 \pmod 2$. Hence the pair of
	assignments $f,g$ cover
	$\calE$.
\end{proof}

\begin{lemma}[Soundness]
	
	\label{lemma:4lin_soundness}
	Let $c_0$ be the constant from \expref{Theorem}{thm:lc-hard}. If
	$G$  is at most $s$-satisfiable
	with  $s < \nicefrac{\delta^{10/c_0 + 5}}{4}$, then any
	independent set in $\cal G$ has fractional size at most $\delta$.
\end{lemma}

\begin{proof}
	Let $I \subseteq \calV$ be an independent set of fractional
	size $\delta$ in the constraint graph $\calG$. For every
	variable $v \in V$, let $f_v : \{0, 1 \}^{2R} \rightarrow
	\{0,1\}$ be the indicator function of the independent set
	restricted to the vertices that correspond to $v$.  Since
	$I$ is an independent set, we have
	\begin{equation}
		\label{eqn:4lin-1}
		\E_{u, v,w}~ \E_{\substack{x,x',\\z,z',\\\eta,\eta',y}}\left[ f_v(x)f_v(x+y\circ \pi_{uv}+\eta
		\circ \pi_{uv} \cdot z)f_w(x')f_w(x'+y\circ\pi_{uw}+\eta'\circ \pi_{uw}\cdot z' + 1) 
		\right] = 0.
	\end{equation}
	For $\alpha \subseteq [2R]$, let $\pi^{\oplus}_{uv}(\alpha) \subseteq [2L]$ be the set
	containing elements $i\in [2L]$ such that if $i <L$ there are
	an odd number of $j\in [R]\cap \alpha$ with $\pi_{uv}(j) = i$
	and if $i \geq L$ there are an odd number of $j \in
	([2R]\setminus [R])\cap \alpha$ with $\pi_{uv}(j-R) = i-L$ .
	It is easy to see that $\chi_\alpha(y\circ \pi_{uw}) =
	\chi_{\pi^\oplus_{uv}(\alpha)} (y)$. Expanding $f_v$ in the Fourier
	basis and taking expectation over $x,x'$ and $y$, we get that
	\begin{equation}
		\E_{u, v,w}\sum_{\alpha,\beta \subseteq [2R]: \pi^\oplus_{uv}(\alpha)=\pi^\oplus_{uw}(\beta)}
		\widehat{f}_v(\alpha)^2\widehat{f}_w(\beta)^2(-1)^{|\beta|} \E_{z,z',\eta,
			\eta'}\left[ \chi_\alpha(\eta\circ \pi_{uv} \cdot z)\chi_\beta(\eta'\circ \pi_{uw}
		\cdot z')\right]=0.
	\end{equation}
	Now the expectation over $z,z'$ simplifies as
	\begin{equation}
		\label{eqn:4lin-2}
		\E_{u, v,w}\sum_{\alpha,\beta \subseteq [2R]: \pi^\oplus_{uv}(\alpha)=\pi^\oplus_{uw}(\beta)}
		\underbrace{\widehat{f}_v(\alpha)^2\widehat{f}_w(\beta)^2(-1)^{|\beta|} \Pr_{\eta,\eta'}[ \alpha 
			\cdot(\eta\circ \pi_{uv}) = \beta\cdot  (\eta'\circ\pi_{uw}) = \bar 0]}_{=:\term_{u,v,w}(\alpha,\beta)}=0,
	\end{equation}
	where we think of $\alpha,\beta$ as the characteristic vectors
	in $\{0,1\}^ {2R}$ of the corresponding sets.  We will now
	break up the above summation into different parts and bound
	each part separately. For a projection $\pi:[R] \rightarrow
	[L]$, define $\widetilde{\pi}(\alpha) := \{ i \in [L]: \exists j 			\in [R],( j \in \alpha \vee j + R \in \alpha) \wedge (\pi(j) = i) \}$.
	We divide the space of $(\alpha,\beta)$ into 4 sets as
	follows.
	\begin{align*}
		E_0 &:= \left\{(\alpha,\beta) \middle|
		\pi^\oplus_{uv}(\alpha)=
		\pi^\oplus_{uw}(\beta)=\emptyset \right\}\, ,\\
		E_1 &:= \left\{(\alpha,\beta) \middle|
		\pi^\oplus_{uv}(\alpha)=
		\pi^\oplus_{uw}(\beta)\neq \emptyset,
		\max\{|\alpha|,|\beta|\} \leq 2/\delta^{5/c_0}
		\right\}\, ,\\
		E_2 &:= \left\{(\alpha,\beta) \middle|
		\pi^\oplus_{uv}(\alpha)= \pi^\oplus_{uw}(\beta)
		\neq\emptyset,
		\max\{|\widetilde{\pi}_{uv}(\alpha)| ,
		|\widetilde{\pi}_{uw}(\beta)|\} \geq  1/\delta^5
		\right\}\, ,\\ 
		E_3 &:= \left\{(\alpha,\beta) \middle|
		\pi^\oplus_{uv}(\alpha)= \pi^\oplus_{uw}(\beta)
		\neq \emptyset, \max\{|\alpha| ,|\beta|\} >
		2/\delta^{5/c_0}, \max
		\{|\widetilde{\pi}_{uv}(\alpha)| ,
		|\widetilde{\pi}_{uw}(\beta)| \} < 1/\delta^5
		\right\}\, .
	\end{align*}
	And define the following quantities for $i \in \{0,1,2,3\}$.
	\[ \Theta_i := \sum_{(\alpha,\beta) \in E_i}
	\E_{u,v,w}\term_{u,v,w}(\alpha, \beta) \, .\]
	
	%
	%
	\paragraph{Lower bound for $\Theta_0$.} If
	$\pi^\oplus_{uw}(\beta) = \emptyset$, then $|\beta|$ is
	even. Hence, all the terms in $\Theta_0$ are positive and 
	$$\Theta_0 \geq  \E_{u,v,w}\term_{u,v,w}(0,0) = \E_u \left(\E_v \widehat 
	f_v(0)^2\right)^2 \geq \left(\E_{u,v} \widehat f_v(0)\right)^4
	= \delta^4.$$
	\paragraph{Upper bound for $\Theta_1$.}
	
	Consider the following strategy for labeling vertices $u\in U$
	and $v\in V$. For $u\in U$, pick a random neighbor $v$, choose
	$\alpha$ with probability $\widehat f_v(\alpha)^2$ and set its
	label to a random element in $\widetilde{\pi}_{uv}( \alpha)$. For $w\in V$,
	choose $\beta$ with probability $\widehat f_w(\beta) ^2$ and
	set its label to a random element of $\beta$. If the label
	$j\geq R$, change the label to $j-R$. The probability that a
	random edge $(u,w)$ of the label cover is satisfied by this
	labeling is
	\begin{align*}
		\label{eqn:fourier-decoc}
		\E_{u,v,w}  \sum_{\substack{\alpha,\beta :\\ \widetilde{\pi}_{uv}(\alpha)\cap \widetilde{\pi}_{uw}(\beta) 
				\neq \emptyset }} \widehat f_v(\alpha)^2 \widehat f_w(\beta)^2 \frac{1}{|\widetilde{\pi}_{uv}
			(\alpha)|\cdot |\beta|} &\geq \E_{u,v,w} \sum_{\substack{\alpha,\beta:\\ \pi
				^\oplus_{uv}(\alpha)= \pi^\oplus_{uw}(\beta) \neq \emptyset\\ \max\{|\alpha|,|\beta|\} 
				\leq 2/\delta^{5/c_0} }} \hspace{-25pt} \widehat f_v(\alpha)^2 \widehat f_w(\beta)^2 
		\frac{\delta^{10/c_0}}{4} \\ &\geq |\Theta_1|\cdot \frac{\delta^{10/c_0}}{4}.
	\end{align*}
	Since the instance is at most $s$-satisfiable, the above is
	not greater than $s$.   
	Choosing $s < \nicefrac{\delta^{10/c_0 + 5}}{4}$, will imply $|\Theta_1| \leq \delta^5$.
	
	\paragraph{Upper bound for $\Theta_2$.}
	
	Suppose $|\widetilde{\pi}_{uv}(\alpha)| \geq 1/\delta^5$, then note that 
	$$\Pr_{\eta,\eta'}[ \alpha \cdot(\eta\circ \pi_{uv}) = \beta\cdot  (\eta'\circ\pi_{uw}) = 0] \leq 
	\Pr_\eta[ \alpha \cdot(\eta\circ \pi_{uv}) = 0] \leq (1-\epsilon)^{|\widetilde{\pi}_{uv}(\alpha)|} \leq (1-\epsilon)^{1/\delta^5}.$$
	Since the sum of squares of Fourier coefficients  of $f$ is less than $1$ and
	$\epsilon$ is a constant, we get that $|\Theta_2| \leq 1/2^{\Omega(1/
		\delta^5)} < O(\delta^5)$.
	
	\paragraph{Upper bound for $\Theta_3$.}
	
	From the third property of \expref{Theorem}{thm:lc-hard}, we have that for any 
	$v\in V$ and $\alpha\subseteq [2R]$ with $|\alpha|> 2/\delta^{5/c_0}$, the probability 
	that $|\widetilde{\pi}_{uv}(\alpha)| < 1/\delta^{5}$, for a random neighbor $u$ of $v$, is at most $\delta^{5}$. Hence $|\Theta
	_3|\leq \delta^5$.
	
	\paragraph{}
	
	On substituting the above bounds in \expeqref{Equation}{eqn:4lin-2},
	we get that $\delta^4 -O(\delta^5) \leq 0$, 
	which gives a contradiction for small enough $\delta$. Hence
	there is no independent set in $\calG$ of size $\delta$.
\end{proof}

\begin{proof}[Proof of {\expref{Theorem}{thm:4lin}}]
	From \expref{Theorem}{thm:lc-hard}, the size of the CSP instance $\calG$ produced 
	by the reduction is $N = n^{r}2^{2^{O(r)}}$ and the parameter $s\leq 2^{-d_0 r}
	$. Setting $r = \Theta(\log \log n)$, gives that $N=2^{\poly \log n}$ and the 
	size of the largest independent set $\delta = 1/\poly\log n= 1/\poly\log
	N$. 
\end{proof}


\section{Invariance Principle for correlated spaces}
\label{sec:inv-prin}

\noindent {\bf \expref{Theorem}{thm:inv-prin} (Invariance Principle for
	correlated spaces) [Restated]} \emph{Let $(\Omega_1^k \times \Omega_2^k, \mu)$ be a
	correlated probability space such that the marginal of $\mu$ on any
	pair of coordinates one each from $\Omega_1$ and $\Omega_2$ is a
	product distribution. Let $\mu_1 ,\mu_2$ be the marginals of $\mu$
	on $\Omega_1^k$ and $\Omega_2^k$, respectively. Let $X, Y$ be two
	random $k\times L$ dimensional matrices chosen as 
	follows.  Independently   
	for every $i \in [L]$, the pair of columns $(x^i,y^i)
	\in \Omega_1^k \times \Omega_2^k$ is chosen from $\mu$. Let
	$x_i,y_i$ denote the $i$-th rows of $X$ and $Y$, respectively.  If
	$F: \Omega_1^L \rightarrow [-1,+1]$ and $G: \Omega_2^L \rightarrow
	[-1,+1]$ are functions such that
	$$\tau:= \sqrt{\sum_{i \in [L]}\Inf_i[F]\cdot \Inf_i[G]}  ~\text{ and } ~
	\Gamma := \max \left\{ \sqrt{\sum_{i \in [L]}\Inf_i[F]} , \sqrt{\sum_{i \in 
			[L]}\Inf_i[G]} \right\} \ ,$$ then
	\begin{equation}
		\abs{ \E_{(X,Y) \in \mu^{\otimes L}} \left[\prod_{i\in [k]}F(x_i) G(y_i)
			\right] - \E_{X \in \mu_1^{\otimes L}} \left[\prod_{i\in [k]}F(x_i)\right]
			\E_{Y \in \mu_2^{\otimes L}} \left[\prod_{i\in [k]}G(y_i)\right] } \leq 
		2^{O(k)} \Gamma \tau.
	\end{equation}
}

\begin{proof}
	
	We will prove the theorem by using the hybrid argument. For $i \in [L+1]$, let 
	$X^{(i)},Y^{(i)}$ be distributed according to $(\mu_1 \otimes \mu_2)^{\otimes
		i} \otimes \mu^{\otimes L- i}$. Thus, $(X^{(0)},Y^{(0)}) =
	(X,Y)$ is distributed according to $\mu^{\otimes L}$ while
	$(X^{(L)},Y^{(L)})$ is distributed according to
	$(\mu_1\otimes\mu_2)^{\otimes L}$. For $i \in [L]$, define
	
	\begin{equation}
		\label{eqn:err}
		\err_i := \abs{ \E_{X^{(i)},Y^{(i)}} \left[\prod_{j=1}^kF(x^{(i)}_j) 
			G(y^{(i)}_j)\right] - \E_{X^{(i+1)},Y^{(i+1)}} \left[\prod_{j=1}^kF(x^
			{(i+1)}_j) G(y^{(i+1)}_j)\right] }.
	\end{equation}
	
	The left-hand side of \expeqref{Equation}{eqn:inv-eqn} is
	not greater than  
	$\sum_{i\in [L]} 
	\err_i$.  Now for a fixed $i$,  we will bound $\err_i$. We use the 
	Efron--Stein decomposition of $F,G$ to split them into two
	parts: the part
	that   
	depends on the 
	$i$-th input and the part independent of the $i$-th input. 
	$$F= F_0 + F_1 \text{ where } F_0 := \sum_{\alpha : i\notin \alpha} F_\alpha \mbox{ and } 
	F_1 := \sum_{\alpha : i\in \alpha} F_\alpha.$$
	$$G = G_0 + G_1 \text{ where } G_0 := \sum_{\beta : i\notin \beta} G_\beta  \mbox{ and } 
	G_1 := \sum_{\beta : i\in \beta} G_\beta.$$
	Note that $\Inf_i[F] = \|F_1\|^2_2$ and $\Inf_i[G] =
	\|G_1\|_2^2$. Furthermore, the functions $F_0$ and $F_1$ are
	bounded since $F_0(x) = \E_{x^{'}} [F(x^{'}) |
	x^{'}_{[L]\setminus i} = x_{[L]\setminus i} ] \in [-1,+1]$ and
	$F_1(x) = F(x) - F_0(x) \in [-2,+2]$.  For $a \in \{0,1\}^k$,
	let $F_a(X) := \prod_{j =1}^kF_{a_j}(x_j)$.  Similarly
	$G_0,G_1$ are bounded and $G_a$ defined
	analogously. Substituting these definitions in 
	\expeqref{Equation}{eqn:err} and expanding the products gives
	$$\err_i = \abs{ \sum_{a,b \in \{0,1\}^k}\left(  \E_{X^{(i)},Y^{(i)}} 
		\left[F_{a}(X^{(i)}) G_{b}(Y^{(i)})\right]  - \E_{X^{(i+1)},Y^{(i+1)}} \left[
		F_{a}(X^{(i+1)}) G_{b}(Y^{(i+1)})\right]  \right) }.$$
	Since both the distributions are identical on
	$(\Omega_1^k)^{\otimes L}$ and $(\Omega_2^k)^{\otimes L}$, all
	terms with $a = \bar 0$ or $b=\bar 0$ are zero. For instance when $a = \bar 0$, $F_a$ does not depend on the $i$-th coordinate. Therefore, in both the distributions $(X^{(i)},Y^{(i)})$ and $(X^{(i+1)},Y^{(i+1)})$, the $i$-th column of $X$ can be dropped. Now, the
	distributions  
	of $Y^{(i)}$ and $Y^{(i+1)}$ are identical conditioned on the $X$ with the $i$-th column dropped. Thus, the expectation is $0$.
	
	Since $\mu$
	is uniform on any pair of coordinates on each from the
	$\Omega_1$ and $\Omega_2$ sides, terms with $|a|=|b|=1$ also
	evaluates to zero using a similar argument as above. Now consider the remaining terms with $ |a|,|b| \geq1,
	|a|+|b| > 2$. Consider one such term where $a_{1},a_2 = 1$ and $b_{1}
	=1$. In this case, by 
	the 
	Cauchy--Schwarz inequality we have that
	\begin{align*}
		\hspace{-13pt}\abs{ \E_{X^{(i-1)},Y^{(i-1)}} \left[F_a(X^{(i-1)}) G_{b}(Y^{(i-1)})\right]}
		&\leq \sqrt{\E F_1(x_1)^2 G_1(y_1)^2} \cdot \|F_1\|_2  \cdot \left\| 
		\prod_{j>2} F_{a_j}\right\|_\infty \hspace{-5pt}\cdot \left\| \prod_{j> 1} G_{b_j}\right
		\|_\infty\hspace{-10pt}.
	\end{align*}
	From the facts that the marginal of $\mu$ to any pair of
	coordinates one each from $\Omega_1$ and $\Omega_2$ sides are
	uniform, $\Inf_i[F] = \|F_1\|_2^2$ and
	$|F_0(x)|,|F_1(x)|,|G_0(x)|,|G_1(x)|$ are all bounded by $2$,
	the right side of above becomes
	\begin{align*}
		\sqrt{\E F_1(x_1)^2} \sqrt{\E G_1(y_1)^2} \cdot \|F_1\|_2  \cdot \left\|
		\prod_{j>2} F_{a_j}\right\|_\infty \cdot \left\| \prod_{j> 1} G_{b_j}
		\right\|_\infty \leq  \sqrt{\Inf_i[F]^2 \Inf_i[G]} \cdot 2^{2k} .
	\end{align*}
	All the other terms corresponding to other
	pairs  
	$(a,b)$, which are  
	at most $2^{2k}$ in number, are bounded analogously. Hence,
	\begin{align*}
		\sum_{i \in [L]} \err_i &\leq 2^{4k} \sum_{i \in [L]} \left( \sqrt{\Inf_i[F]
			^2\Inf_i[G]} +\sqrt{\Inf_i[F]\Inf_i[G]^2} \right)\\
		&= 2^{4k} \sum_{i \in [L]} \sqrt{\Inf_i[F]
			\Inf_i[G]}\left( \sqrt{\Inf_i[F]} +\sqrt{\Inf_i[G]} \right).
	\end{align*}
	Applying the Cauchy--Schwarz inequality, followed by a triangle inequality, we obtain	
	\begin{align*}
		\sum_{i \in [L]} \err_i&\leq 2^{4k} \sqrt{\sum_{i \in [L]} \Inf_i[F]\Inf_i[G]}\left(\sqrt{\sum_{i \in [L]}  \Inf_i[F]} + \sqrt{\sum_{i \in 
				[L]}  \Inf_i[G]} \right).
	\end{align*}
	This completes the proof.   
\end{proof}

{\small
	\bibliographystyle{prahladhurl}
	\bibliography{covering-bib}
}

\appendix

\end{document}

%% file: macros.tex
\newtheorem{theorem}{Theorem}[section]
\newtheorem{conjecture}[theorem]{Conjecture}
\newtheorem{definition}[theorem]{Definition}
\newtheorem{lemma}[theorem]{Lemma}

\newtheorem{proposition}[theorem]{Proposition}
\newtheorem{corollary}[theorem]{Corollary}
\newtheorem{claim}[theorem]{Claim}

\def\FullBox{\hbox{\vrule width 6pt height 6pt depth 0pt}}

\def\qed{\ifmmode\qquad\FullBox\else{\unskip\nobreak\hfil
\penalty50\hskip1em\null\nobreak\hfil\FullBox
\parfillskip=0pt\finalhyphendemerits=0\endgraf}\fi}

\def\qedsketch{\ifmmode\Box\else{\unskip\nobreak\hfil
\penalty50\hskip1em\null\nobreak\hfil$\Box$
\parfillskip=0pt\finalhyphendemerits=0\endgraf}\fi}

\newenvironment{proof}{\begin{trivlist} \item {\bf Proof:~~}}
   {\qed\end{trivlist}}

\newcommand\nat{\mathbb N}

\newcommand\N{\mathbb N}
\newcommand\R{\mathbb R}

\newcommand{\marginlabel}[1]%
{\mbox{}\marginpar{\it{\raggedleft\hspace{0pt}#1}}}
\newcommand\poly{\mbox{poly}}  

\definecolor{Mygray}{gray}{0.8}

 \ifcsname ifcommentflag\endcsname\else
  \expandafter\let\csname ifcommentflag\expandafter\endcsname
                  \csname iffalse\endcsname
\fi

\ifnum\showauthornotes=1
\else
\fi

\ifnum\showauthornotes=1
\newcommand{\Authornote}[2]{{\sf\small\color{red}{[#1: #2]}}}
\newcommand{\Authoredit}[2]{{\sf\small\color{red}{[#1]}\color{blue}{#2}}}
\newcommand{\Authorcomment}[2]{{\sf \small\color{Mygray}{[#1: #2]}}}
\newcommand{\Authorfnote}[2]{\footnote{\color{red}{#1: #2}}}
\newcommand{\Authorfixme}[1]{\Authornote{#1}{\textbf{??}}}
\newcommand{\Authormarginmark}[1]{\marginpar{\textcolor{red}{\fbox{
#1:!}}}}
\else
\newcommand{\Authornote}[2]{}
\newcommand{\Authoredit}[2]{}
\newcommand{\Authorcomment}[2]{}
\newcommand{\Authorfnote}[2]{}
\newcommand{\Authorfixme}[1]{}
\newcommand{\Authormarginmark}[1]{}
\fi

\newcommand\calG{\mathcal{G}}

\newcommand\calD{\mathcal{D}}

\def\abs#1{\left| #1 \right|}

\newlength{\pgmtab}  
 {
	\begin{enumerate}}{\end{enumerate}}

\newcounter{lecnum}

\newlength{\tpush}
\ifnum\showdraftbox=1

\else

\fi